\RequirePackage{fix-cm}
\documentclass[smallextended]{svjour3}       
\smartqed  
\usepackage{graphicx, cancel}
\usepackage[dvips]{geometry}
\usepackage{graphics}
\usepackage{color, soul}
\usepackage{amsmath, slashed}
\usepackage{amssymb}
\usepackage{rotating}
\usepackage{mathrsfs}
\usepackage{epsfig}
\usepackage{verbatim}
\usepackage[affil-it]{authblk}

\usepackage{mathrsfs}
 \usepackage{xcolor}
 \usepackage[citebordercolor={green}]{hyperref}

\newcommand{\bea}{\begin{eqnarray}}
\newcommand{\eea}{\end{eqnarray}}
\newcommand{\beaa}{\begin{eqnarray*}}
\newcommand{\eeaa}{\end{eqnarray*}}
\newcommand{\bsplit}{\begin{split}}

\newcommand{\les}{\lesssim}
\newcommand{\ges}{\gtrsim}
\newcommand{\lot}{\mbox{l.o.t.}}
\newcommand{\dual}{{\,^\star \mkern-2mu}}
\newcommand{\tr}{\mbox{tr}}
\newcommand{\nabb}{\nab\mkern-13mu /\,}

\newcommand{\nab}{\nabla}
\renewcommand{\div}{\mbox{div }}

\newcommand{\divv}{\mbox{div}\mkern-19mu /\,\,\,\,}
\newcommand{\lapp}{\mbox{$\bigtriangleup  \mkern-13mu / \,$}}
\newcommand{\curll}{\mbox{curl}\mkern-19mu /\,\,\,\,}
\newcommand{\pr}{\partial}

\newcommand{\DD}{{\mathcal D}}
\newcommand{\DDs}{ \, \DD \hspace{-2.4pt}\dual    \mkern-16mu /}
\newcommand{\DDd}{ \, \DD \hspace{-2.4pt}    \mkern-8mu /}

\renewcommand{\a}{\alpha}
\renewcommand{\b}{\beta}
\newcommand{\de}{\delta}
\newcommand{\ep}{\epsilon}

\newcommand{\Si}{\Sigma}

\renewcommand{\th}{\theta}

\newcommand{\ka}{\kappa}

\newcommand{\Up}{\Upsilon}

\newcommand{\Lb}{{\underline{L}}}

\newcommand{\chib}{\underline{\chi}}

\newcommand{\kab}{\underline{\kappa}}
\newcommand{\chih}{\widehat{\chi}}

\newcommand{\bF}{\,^{(F)} \hspace{-2.2pt}\b}

\newcommand{\rhoF}{\,^{(F)} \hspace{-2.2pt}\rho}

\newcommand{\EE}{{\mathcal E}}
\newcommand{\HH}{{\mathcal H}}

\newcommand{\LL}{{\mathcal L}}
\newcommand{\MM}{{\mathcal M}}
\newcommand{\PP}{{\mathcal P}}

\newcommand{\TT}{{\mathcal T}}

\newcommand{\F}{{\bf F}}

\newcommand{\g}{{\bf g}}
\newcommand{\pf}{\frak{p}}
\newcommand{\qf}{\frak{q}}
\newcommand{\ff}{\frak{f}}

\newcommand{\piT}{{\,^{(T)} \pi }}

\newcommand{\piX}{\,^{(X)}\pi}
\newcommand{\piY}{\,^{(Y)}\pi}
\newcommand{\piZ}{  \,^{(Z)}\pi}

\renewcommand{\c}{\cdot}

\newcommand{\Mor}{\mbox{Mor}}

\begin{document}

\title{The linear stability of Reissner-Nordstr{\"o}m spacetime: \\
the full subextremal range $|Q|<M$
}

\titlerunning{The linear stability of Reissner-Nordstr{\"o}m spacetime}        

\author{Elena Giorgi       
}


\institute{
              Department of Mathematics, Princeton University \\
              \email{egiorgi@princeton.edu}           
}

\date{Received: date / Accepted: date}

\maketitle

\begin{abstract}
We prove the linear stability of subextremal Reissner-Nordstr\"om spacetimes as solutions to the Einstein--Maxwell equation. We make use of a novel representation of gauge-invariant quantities which satisfy a symmetric system of coupled wave equations. This system is composed of two of the three equations derived in our previous works \cite{Giorgi4}, \cite{Giorgi5}, where the estimates required arbitrary smallness of the charge. Here, the estimates are obtained by defining a combined energy-momentum tensor for the system in terms of the symmetric structure of the right hand sides of the equations. We obtain boundedness of the energy, Morawetz estimates and decay for the full subextremal range $|Q|<M$, completely in physical space. Such decay estimates, together with the estimates for the gauge-dependent quantities of the perturbations obtained in \cite{Giorgi6}, settle the problem of linear stability to gravitational and electromagnetic perturbations of Reissner-Nordstr\"om solution in the full subextremal range $|Q|< M$. 
\end{abstract}

\tableofcontents

   \section{Introduction}
  
  The problem of stability of black holes  as solutions to the Einstein equation occupies a central stage  in mathematical General Relativity \cite{Ch}.
The resolution of this problem consists in understanding the long-time dynamics of perturbations of known stationary solutions to the Einstein equation.

  There are few examples of exact solutions to the Einstein equation, the most fundamental of which is the Kerr spacetime \cite{kerr}, axially symmetric and stationary solution to the Einstein vacuum equation
  \beaa
  \operatorname{Ric}(g)=0,
  \eeaa
  where $g$ is a Lorentzian metric in $3+1$-dimensions. 
  A particular case of the Kerr spacetime is the Schwarzschild solution \cite{Schwarz}, which is spherically symmetric and static, given in coordinates $(t, r, \th, \phi)$ by 
  \beaa
  g_M&=& -\left(1-\frac{2M}{r} \right) dt^2+\left(1-\frac{2M}{r}  \right)^{-1} dr^2+r^2 \left(d\th^2+\sin^2\th d\phi^2\right).
  \eeaa
  The parameter $M$ can be interpreted as the mass of the black hole. 
  
In the case of the Einstein equation coupled with electromagnetic fields, the Lorentzian metric $g$ satisfies the Einstein--Maxwell equation \cite{Chandra}, i.e.
  \bea \label{Einstein-Maxwell-eq}
\operatorname{Ric}(g)_{\mu\nu}&=&2 F_{\mu \lambda} {F_\nu}^{\lambda}- \frac 1 2 g_{\mu\nu} F^{\alpha\beta} F_{\alpha\beta}, \label{Einstein-1}\\
D_{[\alpha} F_{\beta\gamma]}&=&0, \qquad D^\alpha F_{\alpha\beta}=0, \label{Maxwell}
\eea
where $F$ is a two-form verifying the Maxwell equations \eqref{Maxwell}, and $D$ is the Levi-Civita connection of $g$. In this context, the fundamental stationary and axisymmetric solution is the Kerr--Newman spacetime \cite{Newman}, and its spherically symmetric and static case is given by the Reissner--Nordstr\"om metric \cite{Nordstrom}, given in coordinates by 
  \bea\label{RN}
  g_{M, Q}&=& -\left(1-\frac{2M}{r} +\frac{Q^2}{r^2} \right) dt^2+\left(1-\frac{2M}{r}+\frac{Q^2}{r^2}  \right)^{-1} dr^2+r^2 \left(d\th^2+\sin^2\th d\phi^2\right).
  \eea
  The parameter $Q$ can be interpreted as the charge of the black hole for $|Q|<M$, which is referred to as the subextremal range. The case $|Q|=M$ is called the extremal case, while $|Q|>M$ corresponds to a spacetime with naked singularities. Observe that for $Q=0$, the Reissner-Nordstr\"om metric \eqref{RN} reduces to the Schwarzschild one. 
  
  The problem of stability of black hole solutions can be roughly divided into three formulations, each of increasing difficulty, from the formal mode analysis of the linearized equations to the fully non-linear perturbations, passing through the problem of linear stability.
  
  The mode stability consists in formally separating the solutions to the linearized Einstein equation into modes at fixed frequencies, and aims at proving the lack of exponentially growing modes for all metric or curvature components. 
 The study of the mode stability of the Schwarzschild solution was initiated by Regge, Wheeler \cite{Regge-Wheeler} and Zerilli \cite{Zerilli} in metric perturbations (in which the metric of the solution is perturbed), and by Barden and Press \cite{Bardeen-Press} in Newman--Penrose formalism (in which the curvature of the solution, through the Newman--Penrose scalars, is perturbed). See also Dotti \cite{Dotti}.  Chandrasekhar \cite{Chandra} identified a transformation theory in mode decomposition which connects the two approaches in Schwarzschild, which is now referred to as Chandrasekhar transformation. Teukolsky \cite{Teukolsky} extended the equations in Newman-Penrose formalism to the Kerr spacetime, and Whiting \cite{Whiting} proved the mode stability for Kerr spacetime.   The study of mode stability of the Reissner-Nordstr\"om spacetime has been initiated by Moncrief  \cite{Moncrief1}, \cite{Moncrief2}, \cite{Moncrief3}, who obtained the wave equations governing the perturbations in metric perturbations. Chandrasekhar \cite{Chandra1}, \cite{Chandra-RN2} completed the study of the fixed mode perturbations in Newman--Penrose formalism. See also Fern\'andez T\'io--Dotti \cite{Dotti2}.

This weak version of stability is however not sufficient to prove boundedness and decay of the solutions even to the linearized equations.   Indeed, the lack of exponentially growing modes is still consistent with the statement that general perturbations with finite initial energy grow unboundedly in time, because results at the level of individual modes do not imply them for the superposition of infinitely many modes \cite{lectures}.

The resolution of the problem of linear stability consists in proving boundedness and decay for the solutions to the linearized Einstein equation, which does not rely on decomposition in modes but rather on a physical space analysis. 
   The stability of the Schwarzschild solution to the linearized Einstein vacuum equation has been obtained by Dafermos--Holzegel--Rodnianski \cite{DHR} by using curvature perturbations and analysis of the Teukolsky equation. The authors introduced a physical-space version of the Chandrasekhar transformation and crucially used the extensive progress on boundedness and decay results for wave equations on black hole backgrounds (for instance \cite{redshift}, \cite{rp}, \cite{lectures}).  Other proofs have followed: see  Hung--Keller--Wang  \cite{mu-tao} for the proof of the linear stability using metric perturbations, through the analysis of Regge--Wheeler and Zerilli equations. See also Hung \cite{pei-ken}, \cite{pei-ken-2} for a proof in the harmonic gauge  and Johnson \cite{Johnson2} in the generalized harmonic gauge.    
 There have been numerous recent results for the linear stability of Kerr spacetime.  Quantitative decay estimates for the Teukolsky equation in slowly rotating Kerr spacetime have been obtained by Ma \cite{ma2} and Dafermos--Holzegel--Rodnianski \cite{TeukolskyDHR}.  Andersson--B\"ackdahl--Blue--Ma \cite{Kerr-lin1} used the outgoing gauge and H\"afner--Hintz--Vasy \cite{Kerr-lin2} used the wave gauge to prove linear stability of Kerr with small angular momentum.    Decay for solutions to the Maxwell equations in Schwarzschild spacetime have been obtained by Blue \cite{BlueMax} and Pasqualotto \cite{Federico}.

 The ultimate goal of the problem of stability is the study of the dynamics of perturbations of solutions to the fully non-linear Einstein equation. The fully non-linear stability of the Kerr(-Newman) family consists in showing that a small perturbation of a Kerr(-Newman) spacetime which is a solution to the non-linear Einstein(-Maxwell) equation converges to another member of the Kerr(-Newman) family. 
 The only proof of non-linear stability with no symmetry assumptions which is known at this stage, in the asymptotically flat regime, is the global non-linear stability of Minkowski spacetime by Christodolou-Klainerman \cite{Ch-Kl}, which was followed by proofs obtained through different approaches (see \cite{Zipser}, \cite{Lind-Rod}, \cite{Hintz-Vasy-M}).     
See also Zipser \cite{Zipser} for the proof of non-linear stability of Minkowski as solution to the Einstein--Maxwell equation. 
   The first proof of non-linear stability of the Schwarzschild spacetime under the class of symmetry of axially symmetric polarized perturbations was given by Klainerman--Szeftel \cite{stabilitySchwarzschild}. 
 In the presence of a positive cosmological constant, the Kerr--de Sitter and the Kerr--Newman--de Sitter family  with small angular momentum have been proved to be non-linearly stable by Hintz--Vasy \cite{Hintz-Vasy} and by Hintz \cite{Hintz-M} respectively.

  In this paper we solve the problem of linear stability of the Reissner--Nordstr\"om spacetime \eqref{RN} as solution to the linearized Einstein--Maxwell equations \eqref{Einstein-Maxwell-eq} and \eqref{Maxwell}, in the full subextremal range $|Q| <M$. More precisely, we prove boundedness and decay statements for solutions to the linearization of the Einstein--Maxwell equation around a Reissner-Nordstr\"om solution, and the analysis is carried out completely in physical space. Here is a rough version of our main theorem. 
  
    \begin{theorem}\label{rough-version-a}[Linear stability of Reissner-Nordstr\"om spacetime to gravitational and electromagnetic perturbations for $|Q|<M$ (Rough version)]  All solutions to the linearized Einstein--Maxwell equations around a Reissner--Nordstr{\"o}m solution $g_{M, Q}$ for $|Q|<M$ in a certain choice of gauge\footnote{The proof is obtained in Bondi gauge, see \cite{Giorgi6}.}  arising from regular asymptotically flat initial data  {\bf remain uniformly bounded} on the exterior and {\bf decay} to a linearized Kerr--Newman solution.
  \end{theorem}

  Theorem \ref{rough-version-a} represents the final step of a program that was initiated by the author in \cite{Giorgi4}, \cite{Giorgi5} and \cite{Giorgi6} to prove the linear stability of Reissner-Nordstr\"om spacetime to gravitational and electromagnetic perturbations in the full subextremal range. More precisely, the series of works \cite{Giorgi4}, \cite{Giorgi5} and \cite{Giorgi6} amounted to the proof of the linear stability in the case of arbitrarily small charge  $|Q| \ll M$. The main result of this paper is to extend the control of all components of the perturbation to the subextremal range $|Q| < M$.

 \begin{remark}\label{remark-subextremal}
   The subextremal range $|Q|<M$ for which the linear stability holds is expected to be optimal. The estimates as hereby derived make use of the redshift vector field at the horizon, and therefore they are degenerate through the extremal case limit $|Q|=M$. In particular, the same decay estimates obtained in \cite{Giorgi6} are not expected to hold in the extremal case due to the Aretakis instability \cite{extremal-1}, \cite{extremal-2}. Such instability causes the growth of transversal derivatives along the horizon of solutions to the non-linear wave equation \cite{extremal-3}, and it is expected that this phenomenon persists in the linearized gravity. Nevertheless, some weaker version of stability, which takes into account such degeneracy of transversal derivative along the event horizon, could hold in the extremal case, where stability and instability phenomena concur. 
   \end{remark}

     In what follows, we recall the main ideas from the series of works \cite{Giorgi4}, \cite{Giorgi5} and \cite{Giorgi6}. They consist in two parts: 
     \begin{itemize}
     \item     The main system of three wave equations governing the perturbations are derived and analyzed for arbitrarily small charge  in \cite{Giorgi4} and \cite{Giorgi5}. More precisely, two wave equations of spin $\pm2$ (governing the gravitational perturbations) are obtained in \cite{Giorgi4}, and one wave equation of spin $\pm1$ (governing the electromagnetic radiations) is obtained in \cite{Giorgi5}. Those are wave equations for quantities which are invariant\footnote{In a linearization of size $\epsilon$, a quantity is called gauge invariant if it changes quadratically, i.e. by terms of the size $\epsilon^2$, when coordinate transformations of size $\epsilon$ are applied. See \cite{Giorgi6}.  } to coordinate transformations at linear level, which we call \textit{gauge-invariant quantities}. 
     \item The above analysis is used in \cite{Giorgi6} to obtain control for all the components of the perturbations, upon a choice of gauge. In particular, the estimates for the gauge-invariant quantities are used to obtain estimates for the gauge-dependent ones. 
     \end{itemize}
     
In the present paper, we will make use in a fundamental way of the system of equations obtained in \cite{Giorgi4} and \cite{Giorgi5}, and derive from them a new system to obtain control for the gauge-invariant quantities in the full subextremal range. The result in \cite{Giorgi6} will then be applied straightforwardly to obtain control for the gauge-dependent quantities from the new estimates obtained here for the gauge-invariant ones. 

We now recall the main results in  \cite{Giorgi4}, \cite{Giorgi5} and \cite{Giorgi6} which are particularly relevant for this work.

\subsection{The spin $\pm2$ system of equations in \cite{Giorgi4}}

 Suppose that $(\MM, g, F)$ is a solution to the Einstein--Maxwell equation such that the manifold $\MM$ can be foliated by $2$-spheres $S$. In \cite{Giorgi4}, we defined the symmetric traceless $2$-covariant $S$-tensors $\a$ and $\ff$ defined relative to a null frame\footnote{A null frame $\{ e_3, e_4, e_A \}_{A=1,2}$ is such that $g\left(e_3,e_3\right) = 0$,  $g\left(e_4,e_4 \right) = 0$, $ g\left(e_3,e_4\right) = -2$, and $e_A$ are orthogonal to $e_3$ and $e_4$. } $\{ e_3, e_4, e_A \}_{A=1,2}$ as 
\beaa
\a_{AB}&=& W(e_4, e_A, e_4, e_B), \qquad \ff_{AB}= \DDs_2 \bF_{AB} + \rhoF \chih_{AB}
\eeaa
where $W$ is the Weyl curvature, $\DDs_2\bF_{AB}=- D_{(A}\bF_{B)}+ \frac 1 2 g_{AB} \divv \bF$ with $D$ the Levi-Civita connection of $g$, $\bF_A=F(e_A, e_4)$, $\rhoF=\frac 1 2 F(e_3, e_4)$ and $\chih$ is the traceless part of the $S$-tensor $\chi_{AB}=g(D_A e_4, e_B)$. 

The 2-tensors $\a$ and $\ff$ are gauge-invariant, and satisfy a coupled system of Teukolsky-type equations of spin\footnote{The spin $\pm2$ refers to 2-tensors on the sphere. } $\pm2$ \cite{Giorgi4}. The estimates for the Teukolsky equations cannot be obtained directly, but rather through a Chandrasekhar transformation to obtain Regge--Wheeler-type equations. 
We defined the derived quantities $\qf$ and $\qf^\F$ \cite{Giorgi4}
\beaa
\qf&=\frac{1}{\kab}\nabb_3\left(\frac{r}{\kab}\nabb_3(r^3 \kab^2 \a)\right), \qquad \qf^\F = \frac{1}{\kab}\nabb_3(r^3 \kab \ \ff)
\eeaa
where $\kab:=\tr\chib$ is the trace of the second null fundamental form and $\nabb_3$ is the projection of the sphere of the covariant derivative along the incoming null direction $D_{e_3}$.

 The gauge-invariant 2-tensors $\qf$ and $\qf^\F$ satisfy a coupled system of linear wave equations of spin $\pm2$ \cite{Giorgi4} which can be schematically written as  
 \bea
 \Box_{g_{M, Q}}\qf+\tilde{V}_1(r)\ \qf&=&Q \cdot  \Big( b_1(r) \lapp_2 \qf^\F+b_2(r)\pr_r\qf^\F+ b_3(r) \qf^\F+\lot\Big) \label{first}\\
\Box_{g_{M, Q}} \qf^\F+\tilde{V}_2(r) \ \qf^\F &=&Q \cdot  \left(c_1(r) \qf +\lot \right) \label{second}
 \eea
 where $\Box_{g_{M, Q}}$ is the d'Alembertian of the Reissner--Nordstr\"om metric $g_{M, Q}$ applied to 2-tensors, $b_i$, $c_i$, $\tilde{V}_i$ are smooth functions of an area radius function $r$, $\lapp_2$ denotes the Laplacian operator on 2-tensors on the sphere and $\lot$ denotes lower order terms (with respect to differentiability) for $\qf$ and $\qf^\F$.

Estimates for this system are obtained in \cite{Giorgi4} in the case of $|Q| \ll M$ by interpreting the right hand sides of \eqref{first} and \eqref{second} as a perturbation of zero. A careful analysis has to be done at the trapping region in order to absorb the spacetime integrals obtained from the right hand side, but the arbitrary smallness of the charge allows to absorb them into the bulk energies of the left hand side of the equations. Through transport estimates, one can then obtain control for the quantities $\ff$ and $\a$. We refer to \cite{Giorgi4} for more details.

Observe that the first of the equations, i.e. equation \eqref{first}, reduces to the Regge--Wheeler equation used in \cite{DHR} in the case of Schwarzschild (with trivial right hand side).

\subsection{The spin $\pm1$ equation in \cite{Giorgi5}}
 
 In \cite{Giorgi5}, we defined the 1-covariant $S$-tensor $\tilde{\b}$ relative to a null frame $\{ e_3, e_4, e_A \}_{A=1,2}$ as
 \beaa
\tilde{\b}_A:= 2\rhoF \b_A-3\rho \bF_A
\eeaa
where $\bF_A=F(e_A, e_4)$, $\rhoF=\frac 1 2 F(e_3, e_4)$, $\b_A=\frac 1 2 W(e_A, e_4, e_3, e_4)$,  $\rho=\frac 1 4 W(e_3, e_4, e_3, e_4)$. 

The 1-tensor $\tilde{\b}$ is a mixed curvature-electromagnetic component which is gauge-invariant and satisfies a Teukolsky-type equation of spin\footnote{The spin $\pm1$ refers to 1-tensors on the sphere. }  $\pm1$ \cite{Giorgi5}. Also in this case, to obtain the estimates for this equation a Chandrasekhar transformation is applied. We defined the derived quantity $\pf$ \cite{Giorgi5}
 \beaa
\pf&=\frac{1}{\kab} \nabb_3(r^5\kab \ \tilde{\b})
\eeaa
which is shown to satisfy a linear wave equation of spin $\pm1$ \cite{Giorgi5} which can be schematically written as
 \bea
 \Box_{g_{M, Q}}\pf+V_1(r) \ \pf&=&Q \cdot a_1(r) \divv\qf^\F
 \label{third}
 \eea
where $a_1$ and $V_1$ are smooth functions $r$ and $\divv$ is the divergence of a symmetric traceless $2$-tensor on the sphere.  Observe that equation \eqref{third} is coupled to equation \eqref{second} through the presence of $\qf^\F$. 
 
 Estimates for this equation  are obtained in \cite{Giorgi5} in the case of $|Q| \ll M$ by using the control of $\qf^\F$ previously obtained in \cite{Giorgi4}. By decomposing the 1-tensors $\pf$ and $\divv \qf^\F$ in spherical harmonics and projecting equation \eqref{third} to the $\ell=1$ harmonics, the right hand side vanishes, and the equation decouples to a single wave equation for the projection of $\pf$ to the $\ell=1$ mode. Standard techniques for decay of wave equations on black hole backgrounds can then be applied to control $\pf_{\ell=1}$, and by transport estimates $\tilde{\b}_{\ell=1}$. The higher spherical harmonics of $\tilde{\beta}$ are controlled by using a relation between the three quantities of the schematic form \cite{Giorgi5}:
 \bea\label{relation-bb}
Q \c \nabb_3\a &=& d_1(r) \ff+d_2(r) \DDs_2 \tilde{\b}
 \eea
 where $d_i$ are smooth functions of $r$. We refer to \cite{Giorgi5} for more details.

 \subsection{The proof of linear stability for small charge in \cite{Giorgi6}}\label{section-lin-stab}

 The conclusions of \cite{Giorgi4} and \cite{Giorgi5} are the pointwise estimates for the gauge-invariant quantities $\qf$, $\qf^\F$ and $\pf$ for $|Q|\ll M$, and those are the starting point of \cite{Giorgi6}, where such estimates are used to obtain control for \textit{all} the remaining components of the perturbation.  Since the remaining components are gauge-dependent, a careful choice of gauge is needed to show that $\qf$, $\qf^\F$ and $\pf$ control the components of the perturbation, and that in addition their decay is optimal and consistent with non-linear applications.

 In \cite{Giorgi6}, we achieve such a proof with the choice of  outgoing null geodesic, or Bondi, gauge. In this gauge, we made use of residual gauge freedom to define normalizations of scalar functions which allow to obtain integrable transport estimates. We obtain a hierarchy of transport estimates with right hand sides in terms of the known $\qf$, $\qf^\F$ and $\pf$. By integrating along null hypersurfaces, pointwise estimates for all the remaining components can be obtained from the estimates previously obtained for $\qf$, $\qf^\F$ and $\pf$ \cite{Giorgi6}.

 In particular, observe that the estimates for the gauge-dependent quantities in \cite{Giorgi6} do not make use of the smallness of the charge \textit{once the gauge-invariant quantities are controlled}. Since the estimates for $\qf$, $\qf^\F$ and $\pf$ in \cite{Giorgi4} and \cite{Giorgi5} are only valid for $|Q| \ll M$, the final proof of the linear stability of Reissner-Nordstr\"om in \cite{Giorgi6} only holds for arbitrarily small charge. We refer to \cite{Giorgi6} for more details.

   We stress here that, if one were able to extend the pointwise estimates for $\qf$, $\qf^\F$ and $\pf$ to the full subextremal range $|Q|<M$, it would be straightforward to apply the proof of linear stability in \cite{Giorgi6}, which does not use smallness of the charge, to the full subextremal range. More precisely, the results on boundedness and decay for the gauge-invariant quantities in Section 8.5 of \cite{Giorgi6} could be upgraded to hold for $|Q|<M$, and therefore the subsequent control on the gauge-dependent quantities in the following sections of \cite{Giorgi6} would also hold in the full subextremal range.

 \subsection{The mixed spin $\pm1$ and spin $\pm2$ system of equations}
 
 We outline here the main ideas which allow us to extend the result of linear stability of Reissner-Nordstr\"om spacetime from very small charge $|Q| \ll M$ (\cite{Giorgi4}, \cite{Giorgi5}, \cite{Giorgi6}) to the full subextremal range $|Q| < M$. The fundamental step is to introduce a system, which we denote \textit{mixed spin $\pm1$ and spin $\pm2$ Regge--Wheeler system}, which governs the gravitational and electromagnetic perturbations of the Reissner-Nordstr\"om solution and has a symmetric structure, which is favorable in the derivation of the estimates.

 We briefly explain how such a system is obtained from the previously mentioned equations appeared in \cite{Giorgi4} and \cite{Giorgi5}.

 Recall the relation \eqref{relation-bb} between $\nabb_3 \a$, $\ff$ and $ \tilde{\b}$. By taking one derivative in the $\nabb_3$ direction, one derives a relation between $\qf$, $\qf^\F$ and $\ff$ of the schematic form (see \eqref{relations-qf-qfF-pf} for the exact expression):
  \bea\label{relation}
Q \c \qf=d_1(r)\qf^\F+d_2(r)\DDs_2 \pf +\lot
 \eea
  Since equations \eqref{first}, \eqref{second}, \eqref{third} for $\qf$, $\qf^\F$ and $\pf$ are three wave equations for three quantities which are related through the identity \eqref{relation}, it is clear that the above system of three equations is equivalent to a system of two equations, which is to say that one of the equations is redundant. We therefore look for a system of two equations which is equivalent to the system of three equations \eqref{first}, \eqref{second}, \eqref{third}. 
 
 Since $\qf$ and $\qf^\F$ are $2$-tensors and $\pf$ is a $1$-tensor on the sphere, neglecting equation \eqref{third} for $\pf$ would cause the absence of control of the projection to the $\ell=1$ spherical mode of the perturbations\footnote{This was basically the approach of our derivation of the estimates in \cite{Giorgi4}.}. For this reason, we decide to neglect one of the first two equations, more precisely \eqref{first}, the equation for $\qf$, which has the most intricate right hand side. 
 
We then substitute $\qf$ through the relation \eqref{relation} into the wave equation \eqref{second}, and we obtain schematically
 \beaa
 \Box_{g_{M, Q}} \qf^\F+\tilde{V}_2(r) \ \qf^\F &=& c_1(r) \left(Q \cdot \qf +\lot \right) \\
 &=& c_1(r) \left(d_1(r)\qf^\F+d_2(r)\DDs_2 \pf \right) 
 \eeaa
 where the lower order terms, denoted $\lot$ cancel out in the above substitution (see Section \ref{diag-sec} for the precise derivation). One then obtains
 \bea\label{fourth}
 \Box_{g_{M, Q}} \qf^\F+V_2(r) \ \qf^\F  &=& a_2(r) \DDs_2 \pf 
  \eea
  for a new potential $V_2$ and a smooth function $a_2$. 
 Observe that equation \eqref{fourth} is now coupled to equation \eqref{third}. By combining the above wave equations of spin $\pm 2$ and spin $\pm 1$ we obtain a system of two coupled linear wave equations of the following schematic form:
 \beaa
 \begin{cases}
 & \Box_{g_{M, Q}}\pf+V_1(r) \ \pf=Q \cdot a_1(r) \divv\qf^\F\\ 
 &\Box_{g_{M, Q}} \qf^\F+V_2(r) \ \qf^\F  = a_2(r) \DDs_2 \pf
  \end{cases}
 \eeaa
 We call the above system the \textit{mixed spin $\pm1$ and spin $\pm2$ Regge--Wheeler system}. Observe that since $\qf$ is related to $\pf$ and $\qf^\F$ through the relation \eqref{relation}, the above system is equivalent\footnote{It is interesting to observe that the one equation used in \cite{DHR} to prove the linear stability of Schwarzschild can be neglected in Reissner--Nordstr\"om in favor of the two equations above (which have no correspondence in the gravitational perturbations of Schwarzschild).} to the system of three equations obtained in \cite{Giorgi4} and \cite{Giorgi5}.  The two quantities $\qf^\F$ and $\pf$ therefore play the role of gravitational and electromagnetic radiation respectively for perturbations of Reissner--Nordstr\"om spacetime.

 The fundamental advantage of the derived system compared to the previous one is in its symmetry: the operators $\DDs_2$ and $\divv$ appearing on the right hand sides are adjoint operators on the sphere \cite{DHR}. 
  Such symmetry is used here to define a combined energy-momentum tensor which allows for a cancellation of the highest order terms, without recurring to smallness of the charge. We are therefore able to deduce boundedness of the energy, Morawetz and $r^p$-estimates in the full subextremal range $|Q| <M$. This is in contrast with the system analyzed in \cite{Giorgi4}, which has non symmetric right hand sides, and for which the analysis can be obtained for very small $Q$ only.

\begin{remark}
 A similar structure in the coupling terms of a system of two wave equations has been found by Hung \cite{pei-ken}, for odd perturbations of linearized gravity of Schwarzschild in harmonic gauge. In \cite{pei-ken}, two metric components, denoted $H_1$ and $H_2$, satisfy a system of wave equations which are coupled through adjoint operators on the sphere, similarly to our mixed spin $\pm1$ and spin $\pm2$ Regge--Wheeler system. Hung obtains estimates for the system through a novel definition of energy-momentum tensor which makes use of the symmetric right hand side. We take a similar approach through the definition of a combined energy-momentum tensor as explained below. 
\end{remark}

 \subsection{The combined energy-momentum tensor}
 
 We now give a brief summary of the proof of boundedness and decay statements for the mixed spin $\pm1$ and spin $\pm2$ Regge--Wheeler system.

 We define a combined energy-momentum tensor for the system, which takes into account both equations and their structure. Such combined energy-momentum tensor $\TT_{\mu\nu}[\qf^\F, \pf]$ is tailored on the specific structure of right hand side of the system. More precisely, it consists of the sum of the energy-momentum tensor associated to each equation, plus a mixed term defined in terms of the right hand side. Schematically:
 \beaa
 \TT_{\mu\nu}[\qf^\F, \pf]&:=& \TT_{\mu\nu}[\qf^\F]+\TT_{\mu\nu}[\pf] -  Q \c a(r)   \left(\DDs_2 \pf \c \qf^\F\right) g_{\mu\nu}
 \eeaa
 where $\TT_{\mu\nu}[\qf^\F]$ and $\TT_{\mu\nu}[\pf]$ are the standard energy-momentum tensor associated to the wave equations for $\qf^\F$ and $\pf$ respectively. See Definition \ref{def-str} for the exact expression. 
 
 Such definition of $\TT_{\mu\nu}[\qf^\F, \pf]$ is motivated by the following property: when applied with multiplier $X=\partial_t$, the associated current $\PP^{X}_{\mu}=\TT_{\mu\nu}[\qf^\F, \pf]X^\nu$ is divergence free. In particular, the additional term $- Q \c a(r) \left(\DDs_2 \pf \c \qf^\F\right) g_{\mu\nu}$ in the definition of the combined energy-momentum tensor is precisely the one needed to obtain cancellation of the divergence. By applying the divergence theorem to a causal domain, one only needs to prove the positivity of the modified boundary terms to obtain boundedness of the energy, which can be obtained in the full subextremal range $|Q|<M$. This is done in Section \ref{sec-en}. 
  
 The derivation of Morawetz estimates is more subtle since the divergence of the current associated to $Y=f(r) \partial_{r}$ does not vanish, but has to be proved to be positive definite.  Because of the mixed term in the definition of $\TT_{\mu\nu}[\qf^\F, \pf]$, we obtain a spacetime integral containing terms of the schematic form
 \beaa
 c_1(r) |\qf^\F|^2+ c_2(r) |\pf|^2 - c_3(r) \left(\qf^\F \c \pf\right)
 \eeaa
 which has to be proved to be positive definite for a well-chosen function $f(r)$. The negativity of the discriminant of the above quadratic form ($D=c_3(r)^2-4c_1(r)c_2(r)$), together with the positivity of the coefficients $c_1(r)$ and $c_2(r)$, is used to conclude that a spacetime integral of the above schematic form is positive definite. This is done in Section \ref{sec-mo}. 
 
 Finally, the $r$-weights appearing on the right hand side of the equations in the mixed spin $\pm1$ and spin $\pm 2$ Regge--Wheeler system are sufficiently good so that the derivation of the $r^p$-hierarchy of Dafermos--Rodnianski for the system is identical to the standard wave equation.  This is done in Section \ref{rp-sec}. 
 
 A rough version of the result is as follows. For the precise version, see Theorem \ref{main-theorem}.
 
 \begin{theorem}\label{rough-version-b}[Rough version] Solutions to the mixed spin $\pm1$ and spin $\pm2$ Regge--Wheeler system on Reissner-Nordstr\"om spacetime with $|Q|<M$ arising from initial data which is prescribed on a Cauchy hypersurface $\Sigma_0$ satisfy statements of energy boundedness, integrated local energy decay, and a hierarchy of $r$-weighted energy estimates.
  \end{theorem}

The hierarchy of $r$-weighted estimates is such that, using a pigeonhole principle \cite{rp}, one obtains in the full subextremal range $|Q|<M$, the pointwise decay estimates 
 \beaa
 | \pf | \leq C \tau^{-1+\de}, \qquad |\qf^\F| \leq C \tau^{-1+\de}
 \eeaa
for $\de>0$ and a time function $\tau$,  where $C$ is some constant depending on an appropriate Sobolev norm of the data. 
 
 The pointwise estimates for $\pf$ and $\qf^\F$ imply estimates for $\qf$ in the full subextremal range through the relation \eqref{relation}. We are then in the condition of having extended the estimates for $\qf$, $\qf^\F$ and $\pf$ to the full subextremal range $|Q|<M$, and therefore the proof of linear stability \cite{Giorgi6} can be applied to obtain control of all the remaining gauge-dependent quantities, as explained in Section \ref{section-lin-stab}.

 The paper is organized as follows. In Section \ref{sec-RN} we recall the main properties of Reissner-Nordstr\"om spacetime and in Section \ref{diag-sec}, the symmetric system used in this paper is derived from the equations obtained in \cite{Giorgi4} and \cite{Giorgi5}. In Section \ref{en-qu}, the energy quantities are defined and the main theorem is stated. The energy-momentum tensor associated to the system is defined in Section \ref{str-sy}. In Section \ref{sec-en}, boundedness of the energy for the full subextremal range $|Q|<M$ is proved. Morawetz estimates for the subextremal range $|Q|< M$ are obtained in Section \ref{sec-mo} and the $r^p$-estimates are derived in Section \ref{rp-sec}.

 \section{The Reissner-Nordstr\"om spacetime}\label{sec-RN}
 
 In this section, we introduce the Reissner-Nordstr{\"o}m exterior metric, as well as relevant background structure. We mostly highlight the properties which are needed in this paper. For a more complete description of the Reissner-Nordstr\"om spacetime see \cite{Exact}.

 \subsection{The manifold and the metric}

Define the manifold with boundary
\begin{align} \label{SchwSchmfld}
\mathcal{M} := \mathcal{D} \times S^2 := \left(-\infty,0\right] \times \left(0,\infty\right) \times S^2
\end{align}
with Kruskal coordinates $\left(U,V,\theta^1,\theta^2\right)$, as defined in Section 3 of \cite{Giorgi4}.
 The boundary $\mathcal{H}^+$ will be referred to as the \emph{horizon}. 
We denote by $S^2_{U,V}$ the $2$-sphere $\left\{U,V\right\} \times S^2 \subset \mathcal{M}$ in $\mathcal{M}$.

Fix two parameters $M>0$ and $Q$, verifying $|Q|<M$. Then the Reissner-Nordstr{\"o}m metric $g_{M, Q}$ with parameters $M$ and $Q$ is defined to be the metric:
\begin{align} \label{sskruskal}
g_{M, Q} = -4 \Up_K \left(U,V\right) d{U} d{V} +   r^2 \left(U,V\right) \gamma_{AB} d{\theta}^A d{\theta}^B.
\end{align}
where 
\beaa
\Up_K \left(U,V\right) &=& \frac{r_{-}r_{+}}{4r(U,V)^2} \Big( \frac{r(U,V)-r_{-}}{r_{-}}\Big)^{1+\left(\frac{r_{-}}{r_{+}}\right)^2}\exp\Big(-\frac{r_{+}-r_{-}}{r_{+}^2} r(U,V)\Big) \\
 \gamma_{AB} &=& \textrm{standard metric on $S^2$} \, .
\eeaa
and 
\bea\label{definiion-rpm}
r_{\pm}=M\pm \sqrt{M^2-Q^2}
\eea
and $r$ is an implicit function of the coordinates $U$ and $V$. We denote $r_{\mathcal{H}}=r_{+}=M+\sqrt{M^2-Q^2}$.

\begin{figure}[h]\label{figure1}
\centering
\def\svgwidth{0.35\textwidth} 
\includegraphics{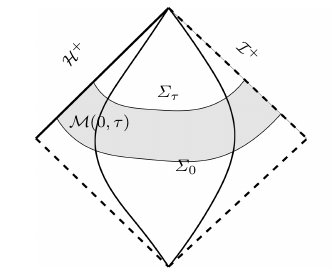}
\caption{Penrose diagram of the patch covered by the $U$ and $V$ coordinates.}
\end{figure}

The Kruskal coordinates cover the entire exterior region up to the horizon. 
 We now define another double null coordinate system that covers  the interior of $\mathcal{M}$ (up to the event horizon), modulo the degeneration of the 
angular coordinates. This coordinate system, 
$\left(u,v,\theta^1, \theta^2\right)$, is called  \emph{double null coordinates} and are defined via the relations
\begin{align} \label{UuVv}
U = -\frac{2r_{+}^2}{r_{+}-r_{-}}\exp\left(-\frac{r_{+}-r_{-}}{4r_{+}^2} u\right) \ \ \ \textrm{and} \ \ \ \ V = \frac{2r_{+}^2}{r_{+}-r_{-}}\exp \left(\frac{r_{+}-r_{-}}{4r_{+}^2} v\right) \, .
\end{align}
Using (\ref{UuVv}), we obtain the Reissner-Nordstr{\"o}m metric on the interior of $\MM$ in $\left(u,v,\theta^1, \theta^2\right)$-coordinates:
\begin{align} \label{ssef}
g_{M, Q} =  - 4 \Up \left(u,v\right) \, d{u} \, d{v} +   r^2 \left(u,v\right) \gamma_{AB} d{\theta}^A d{\theta}^B  
\end{align}
with
\begin{align}
\label{officialOmegadef}
\Up:= 1-\frac{2M}{r}+\frac{Q^2}{r^2}
\end{align}
  We denote by $S_{u, v}$ the sphere $S^2_{U, V}$ where $U$ and $V$ are given by \eqref{UuVv}.

Note that $u, v$ are  regular optical functions. Their corresponding   null geodesic generators are
\bea\label{definition-L-Lb}
\Lb:=-g^{ab}\pr_a v  \pr_b=\frac{1}{\Up} \pr_u,         \qquad  L:=-g^{ab}\pr_a u  \pr_b=\frac{1}{\Up} \pr_v,
\eea

The          null frame $(e_3^*, e_4^*) $        for which $e_3$ is geodesic  (which is regular towards the future along the event horizon) is given by
\bea
\label{eq:regular-nullpair}
e_3^*=\Lb,\,\, \qquad e_4^*=\Up L.
\eea

The null frame $(e_3, e_4)$  for which $e_4$ is geodesic (which is regular towards null infinity) is given by     
\bea\label{outgoing-null-pair}
e_3=\Up \Lb,\,\, \qquad  e_4 =L
\eea

The photon sphere of Reissner-Nordstrom corresponds to the hypersurface in which null geodesics are trapped. It is the hypersurface given by $\{ r=r_P \}$ where $r_P$ is the largest root of the polynomial $r^2-3M r +2Q^2$, and is given by 
\bea\label{def-rP}
r_P=\frac{3M+\sqrt{9M^2-8Q^2}}{2}
\eea
The curvature and electromagnetic components which are non-vanishing are given by 
\bea\label{value-rho}
\begin{split}
\rhoF&=\frac 1 2 F(e_3, e_4)=\frac{Q}{r^2}, \\
 \rho  &=\frac 1 4 W(e_3, e_4, e_3, e_4)=  -\frac{2M}{r^3}+\frac{2Q^2}{r^4} 
 \end{split}
\eea
where $F$ is the electromagnetic tensor and $W$ is the Weyl curvature of the Reissner-Nordstr\"om solution.

\subsection{The Killing vector fields}

We define the vectorfield $T$ to be the timelike Killing vector field $\pr_t$ of the $(t, r)$ coordinates in \eqref{RN}, which in double null coordinates is given by 
\beaa
T=\frac 1 2 (\pr_u+\pr_v)=\frac 12 (\Up e_3^*+e_4^*)=\frac 1 2 (e_3+\Up e_4)
\eeaa
We can also define a basis of angular momentum operator $\Omega_i$, $i=1,2,3$ (see for example \cite{DHR}). The Lie algebra of Killing vector fields of $g_{M,Q}$ is then generated by $T$ and $\Omega_i$, for $i=1,2,3$.

\subsection{The spherical harmonics and elliptic estimates}\label{spherical-harmonics}

We collect some known definitions and properties of the Hodge decomposition of scalars, one forms and symmetric traceless two tensors in spherical harmonics. We also recall some known elliptic estimates. See Section 4.4 of \cite{DHR} for more details.

We denote by $Y_m^\ell$, with $|m|\leq \ell$, the spherical harmonics on the sphere of radius $r$, i.e. 
\bea\label{sph-arm}
\lapp_0 Y^\ell_m=-\frac{1}{r^2} \ell(\ell+1) Y^\ell_m
\eea
 where $\lapp_0$ denotes the laplacian on the sphere $S_{u,v}$ of radius $r=r(u, v)$ for scalar functions. 

\begin{definition}\label{lemma-spherical-harmonics} We say that  a function $f$ on $\mathcal{M}$ is supported on $\ell\geq 2$ if the projections 
\beaa
\int_{S_{u, v}}  f \c Y^\ell_m=0
\eeaa
vanish for $Y_m^{\ell=1}$ for $m=-1, 0, 1$. 
\end{definition}

We recall the following angular operators on $S_{u,v}$-tensors. 
Let $\xi$ be an arbitrary one-form and $\th$ an arbitrary symmetric traceless $2$-tensor on $S_{u,v}$. 
\begin{itemize}
\item $\nabb$ denotes the covariant derivative associated to the metric $\slashed{g}$ on $S_{u,v}$.
\item $\DDd_1$ takes $\xi$ into the pair of functions $(\divv \xi, \curll \xi)$, where $$\divv \xi=\slashed{g}^{AB} \nabb_A \xi_B, \qquad \curll\xi=\slashed{\ep}^{AB}\nabb_A \xi_B$$
\item $\DDs_1$ is the formal $L^2$-adjoint of $\DDd_1$, and takes any pair of functions $(\rho, \sigma)$ into the one-form $-\nabb_A \rho+\slashed{\ep}_{AB} \nabb^B \sigma$.
\item  $\DDd_2$ takes $\th$ into the one-form $\DDd_2\th=(\divv \th)_C=\slashed{g}^{AB}\nabb_A \th_{BC}$.
\item $\DDs_2$ is the formal $L^2$-adjoint of $\DDd_2$, and takes $\xi$ into the symmetric traceless two tensor $$(\DDs_2\xi)_{AB}=-\frac 12 \left( \nabb_B\xi_A+\nabb_A\xi_B-(\divv \xi)\slashed{g}_{AB}\right)$$
\end{itemize}

We can easily check that $\DDs_k$ is the formal adjoint of $\DDd_k$, i.e.
\bea\label{adjoint}
\int_{S} (\DDd_k f) g =\int_{S} f (\DDs_k g )
\eea

Recall that an arbitrary one-form $\xi$ on $S_{u,v}$ has a unique representation $\xi=r \DDs_1(f, g)$, where $\DDs_1(f, g)$, for two uniquely defined functions $f$ and $g$ on the unit sphere, both with vanishing mean. In particular, the scalars $\divv \xi$ and $\curll \xi$ are supported in $\ell\ge1$. 

Recall that an arbitrary symmetric traceless two-tensor $\th$ on $S_{u,v}$ has a unique representation $\th=r^2\DDs_2\DDs_1(f, g)$ for two uniquely defined functions $f$ and $g$ on the unit sphere, both supported in $\ell\ge 2$. In particular, the scalars $\divv \divv \th$ and $\curll \divv \th$ are supported in $\ell\ge 2$.

We now derive the decomposition in spherical harmonics for one-forms and for two-tensors. We denote $\lapp_1$ and $\lapp_2$ the laplacian on the sphere $S_{u,v}$ of radius $r=r(u, v)$ for one-forms and two-tensors respectively. The laplacian is related to the angular Hodge operators     by the following relations \cite{Ch-Kl}: 
\beaa
\DDd_1 \DDs_1=-\lapp_0, \qquad \DDs_1 \DDd_1= -\lapp_1+K, \\
\DDd_2 \DDs_2= -\frac 1 2 \lapp_1-\frac 1 2 K, \qquad \DDs_2 \DDd_2= -\frac 1 2 \lapp_2+K
\eeaa
Using the above one can prove the following commutators (see Appendix in \cite{Giorgi4}):
 \bea\label{commutators}
 \begin{split}
 -\DDs_1 \lapp_0+\lapp_1\DDs_1&=& K \DDs_1, \\
 -\DDs_2 \lapp_1+\lapp_2 \DDs_2&=&3K \DDs_2 
\end{split}
 \eea
 where $K=\frac{1}{r^2}$ is the Gauss curvature of the sphere of radius $r$. 

Let $\xi$ be a one-form supported on the spherical harmonic $\ell \geq 1$, i.e. $\xi =r \DDs_1(f, g)$ with $f$ and $g$ scalar functions supported on $\ell \geq 1$. We then have from  \eqref{sph-arm} and \eqref{commutators}:
\beaa
\lapp_1 \xi&=&r \lapp_1 \DDs_1(f, g)=r  \DDs_1\lapp_0(f, g)+ K  r\DDs_1(f, g)\\
&=&-\frac{1}{r^2} \ell(\ell+1) r\DDs_1(f, g)+ K r\DDs_1(f, g)\\
&=&-\frac{\ell(\ell+1)-1}{r^2} \xi  
\eeaa
Multiplying the above by $\xi$ and integrating by parts the left hand side, we obtain for a one-form $\xi$ supported on the spherical harmonics $\ell \geq 1$:
\bea\label{elliptic-spj}
\int_S |\nabb \xi|^2&=&\int_S \frac{\ell(\ell+1)-1}{r^2} |\xi|^2
\eea

 Let $\th$ be a symmetric traceless two-tensor supported on the spherical harmonic $\ell \geq 2$, i.e. $\th=r^2\DDs_2\DDs_1(f, g)=r \DDs_2 \xi$ with $\xi$ supported on $\ell \geq 2$.  From \eqref{sph-arm} and \eqref{commutators}, we have
 \beaa
 \lapp_2 \th&=&r \lapp_2 \DDs_2\xi =r  \DDs_2 \lapp_1\xi + 3K r\DDs_2\xi\\
&=&-\frac{\ell(\ell+1)-1}{r^2} r\DDs_2\xi+ 3K r\DDs_2\xi \\
&=&-\frac{\ell(\ell+1)-4}{r^2} \th 
 \eeaa
 Multiplying the above by $\th$ and integrating by parts the left hand side, we obtain for a symmetric traceless two-tensor supported on the spherical harmonic $\ell \geq 2$:
\bea\label{elliptic-spj-2}
\int_S |\nabb \th|^2&=&\int_S \frac{\ell(\ell+1)-4}{r^2} |\th|^2
\eea

We recall the following $L^2$ elliptic estimates. 

\begin{proposition}[\cite{Ch-Kl}] Let $(S, \gamma)$ be a compact surface with Gauss curvature $K$. Then the following identities hold for $1$-forms $\xi$ on $S$ and symmetric traceless $2$-tensors $\th$:
\bea
\int_{S} |\nabb \xi|^2-K |\xi|^2 &=& 2\int_{S} |\DDs_2 \xi|^2 \label{elliptic-one-form} \\
\int_S |\nabb\th|^2+2K |\th|^2&=& 2 \int_S |\DDd_2 \th|^2 \label{elliptic-2-tensor}
\eea
\end{proposition}

We can specialize the above elliptic estimates to tensor supported on a fixed spherical harmonic $\ell$. 
Using \eqref{elliptic-spj} and \eqref{elliptic-one-form}, we deduce for a one form supported on a fixed spherical harmonic $\ell$:
\beaa
\int_S |\DDs_2 \xi|^2&=& \int_S \frac 1 2 |\nabb \xi|^2 -\frac 1 2 K |\xi|^2\\
 &=&\int_S \frac 1 2\frac{\ell(\ell+1)-1}{r^2} |\xi|^2  -\frac 1 2 \frac{1}{r^2} |\xi|^2\\
 &=&\int_S \frac 1 4\frac{2\ell(\ell+1)-4}{r^2} |\xi|^2  
\eeaa
This implies, for one form $\xi$ and two tensor $\th$ both supported on a fixed spherical harmonic $\ell$:
\bea\label{dot-product-estimate-l}
\int_S \DDs_2 \xi \c \th &\geq& -\int_S |\DDs_2 \xi | |\th| = - \int_S \frac 1 2 \frac{(2\ell(\ell+1)-4)^{1/2}}{r} |\xi | |\th|
\eea

\section{The derivation of the mixed spin $\pm1$ and spin $\pm2$ system of equations}\label{diag-sec}

In this section, we derive the mixed spin $\pm1$ and spin $\pm2$ system of equations from the spin $\pm2$ equations obtained in \cite{Giorgi4} and the spin $\pm1$ equation obtained in \cite{Giorgi5}.

Recall the quantities $\pf$, $\qf^\F$ and $\qf$ for linear perturbations of Reissner-Nordstr\"om spacetime as defined in the Introduction.

 In \cite{Giorgi4}, the following wave equation for $\qf^\F$, coupled with $\qf$, has been derived (see Proposition 16 Appendix B.1 in \cite{Giorgi4}):
\bea\label{equation-qf-F-1}
\Box_{g_{M, Q}} \qf^\F+\left( \ka\kab+3 \rho \right)\ \qf^\F &=&\rhoF \left(-\frac 1 r  \qf+4\rhoF \ r^3 \kab \ \ff \right)
\eea
 where here $\Box_{g_{M, Q}}=D^\mu D_\mu$ is the d'Alembertian of the Reissner--Nordstr\"om metric $g_{M, Q}$ applied to 2-tensors. Being an equation for  the linearized quantity $\qf$, the coefficients of the equations are the background values in Reissner-Nordstr\"om. More precisely
  \begin{itemize}
 \item $\ka:=\tr \chi$ and $\kab:=\tr\chib$ are the trace of the second null fundamental forms and $\ka\kab=- \frac{4}{r^2}\left(1-\frac{2M}{r}+\frac{Q^2}{r^2} \right)$, 
 \item $\rho= -\frac{2M}{r^3}+\frac{2Q^2}{r^4} $ and $\rhoF=\frac{Q}{r^2}$ are given by \eqref{value-rho}.
 \end{itemize}
 
 In \cite{Giorgi5},  the wave equation for $\pf$, coupled with $\qf^\F$, has been derived (see Proposition B.1 Appendix B in \cite{Giorgi5}):
\bea\label{equation-pf-1}
\Box_{g_{M, Q}}\pf+\left(\frac 1 4 \ka\kab-5\rhoF^2\right)\pf&=&8r^2 \rhoF^2\divv(\qf^\F)
\eea
where here $\Box_{g_{M, Q}}=D^\mu D_\mu$ is the d'Alembertian of the Reissner--Nordstr\"om metric $g_{M, Q}$ applied to 1-tensors. As above, the coefficients of the equation are the background values in Reissner-Nordstr\"om.

In \cite{Giorgi5}, the following relation between $\tilde{\b}$, $\a$ and $\ff$ has been derived (see Lemma 6.3.1 in \cite{Giorgi5}):
 \beaa
\DDs_2( r^3 \kab \tilde{\b})&=& -\rhoF\frac{1}{\kab}\nab_3(r^3 \kab^2 \a)-\left(2\rhoF^2+3\rho \right)r^3\kab \ff
 \eeaa
 where $\kab$, $\rhoF$ and $\rho$ are again the background values.

We multiply the above by $r^3$, apply  $\frac{1}{\kab} \nabb_3$, and recall that $[\nabb_3, r \DDs_2]=0$ \cite{Giorgi4}. We then obtain
  \beaa
 r\DDs_2 \pf &=& -r^2\rhoF \qf -r^3\left(2\rhoF^2+3\rho \right)\qf^\F\\
&&  -\frac{1}{\kab} \nabb_3\left(r^2\rhoF\right) r\frac{1}{\kab}\nab_3(r^3 \kab^2 \a)-\frac{1}{\kab} \nabb_3\left(r^3(2\rhoF^2+3\rho)\right)r^3\kab \ff
 \eeaa
where we recall that $\qf=\frac{1}{\kab}\nabb_3\left(\frac{r}{\kab}\nabb_3(r^3 \kab^2 \a)\right)$, $\qf^\F = \frac{1}{\kab}\nabb_3(r^3 \kab \ \ff)$ and $\pf=\frac{1}{\kab} \nabb_3(r^5\kab \ \tilde{\b})$. 
 Using that $r^2\rhoF=Q$, and $r^3(2\rhoF^2+3\rho)=-6M+\frac{8Q^2}{r} $,  we obtain
   \beaa
 \DDs_2 \pf &=& -r\rhoF \qf -r^2\left(2\rhoF^2+3\rho \right)\qf^\F +4Q^2 r \kab^2 \ff
 \eeaa
By writing $Q^2 r=r^5 \rhoF^2$, we proved that the quantities $\pf$, $\qf$ and $\qf^\F$ are related through the following relation:
\bea\label{relations-qf-qfF-pf}
\DDs_2\pf&=&  -r\rhoF\qf-r^2(3\rho+2\rhoF^2)\qf^\F +4r^5\rhoF^2 \kab \ff
\eea
We use relation \eqref{relations-qf-qfF-pf} to substitute $\qf$ in \eqref{equation-qf-F-1}:
\beaa
-\frac 1 r \rhoF\qf&=& \frac{1}{r^2}\DDs_2\pf -4r^3\rhoF^2 \kab \ff+(3\rho+2\rhoF^2)\qf^\F 
\eeaa
Equation \eqref{equation-qf-F-1} then becomes
\beaa
\Box_{g_{M, Q}} \qf^\F+\left( \ka\kab+3 \rho \right)\ \qf^\F &=&-\frac 1 r \rhoF \qf+4r^3\rhoF^2 \kab \ \ff \\
&=& \frac{1}{r^2}\DDs_2\pf -4r^3\rhoF^2 \kab \ \ff +(3\rho+2\rhoF^2)\qf^\F+4r^3\rhoF^2 \kab \ \ff  \\
&=& \frac{1}{r^2}\DDs_2\pf +(3\rho+2\rhoF^2)\qf^\F
\eeaa
where observe the cancellation of the term $4r^3\rhoF^2 \kab \ \ff$.
We therefore obtain
\beaa
\Box_{g_{M, Q}} \qf^\F+\left( \ka\kab-2\rhoF^2 \right)\ \qf^\F &=& \frac{1}{r^2}\DDs_2\pf 
\eeaa
We now combine the above equation together with equation \eqref{equation-pf-1} for $\pf$, and we obtain the following system, which we denote \textit{mixed spin $\pm1$ and spin $\pm2$ Regge--Wheeler system}: 
\bea
\begin{split}
\Box_{g_{M, Q}}\pf-V_1(r) \ \pf&=\frac{8Q^2}{r^2}\DDd_2\qf^\F  \label{box-pf}\\
\Box_{g_{M, Q}} \qf^\F-V_2(r) \  \qf^\F &=\frac{1}{r^2}\DDs_2\pf \label{box-qfF}
\end{split}
\eea
where we wrote the divergence as $\div=\DDd_2$. The potentials are given by
\bea
V_1(r)&=&-\frac 1 4 \ka\kab+5\rhoF^2=\frac 1 4  \frac{4}{r^2}\left(1-\frac{2M}{r}+\frac{Q^2}{r^2} \right)+5\frac{Q^2}{r^4}=\frac{1}{r^2}\left(1-\frac{2M}{r}+\frac{6Q^2}{r^2} \right)\label{potential-1} \\ 
V_2(r)&=&-\ka\kab+2\rhoF^2= \frac{4}{r^2}\left(1-\frac{2M}{r}+\frac{Q^2}{r^2} \right)+2\frac{Q^2}{r^4}=\frac{4}{r^2}\left( 1-\frac{2M}{r}+\frac{3Q^2}{2r^2}\right) \label{potential-2}
\eea
In order to make the right hand sides symmetric in the presence of $Q$, we can assume\footnote{This case is contained in the case of $|Q| \ll M$ treated in \cite{Giorgi4} and \cite{Giorgi5}.}  that $Q \neq 0$ and define 
\bea\label{def-phis}
\Phi_1:=\pf, \qquad \Phi_2:=Q \c \qf^\F
\eea
with $\Phi_1$ a 1-tensor and $\Phi_2$ a symmetric traceless 2-tensor on the sphere.
Then the above system becomes
\bea
\Box_{g_{M, Q}}\Phi_1-V_1(r)  \ \Phi_1&=&\frac{8Q}{r^2}\DDd_2 \Phi_2 \label{box-phi-1}\\
\Box_{g_{M, Q}} \Phi_2-V_2(r) \  \Phi_2 &=&\frac{Q}{r^2}\DDs_2\Phi_1 \label{box-phi-2}
\eea
The above two equations form the symmetric system which we will analyze below. 

Observe that we can restrict our attention to the case of $\Phi_1$ supported to the $\ell \geq 2$ spherical harmonics. Indeed, if $\Phi_1$ is supported on the $\ell=1$ spherical harmonics, the two equations decouple since $(\DDd_2 \Phi_2)_{\ell=1}=0$ and $\DDs_2( \Phi_1)_{\ell=1}=0$. More precisely, the first equation reduces to the main equation analyzed in \cite{Giorgi5}, and the second equation reduces to one of the two equations analyzed in \cite{Giorgi4}, with trivial right hand side. In what follows, we will therefore restrict to the case of $\Phi_1$ and $\Phi_2$ both supported to the $\ell \geq 2$ spherical harmonics.

\section{Energy quantities and statements of the main theorem}\label{en-qu}

We define a foliation in Reissner--Nordstr\"om spacetime $\Sigma_\tau$ which connects the event horizon and future null infinity.                 
                   We foliate  $\MM$    by  hypersurfaces  $\Si_\tau $ which are:
             \begin{enumerate}
             \item  Incoming null    in $\{ r_{\mathcal{H}} \leq r \leq \frac{11}{10}r_{\mathcal{H}} \}$,        
              with $e_3^*$ as null incoming   generator (which is regular up to horizon). We denote  this portion 
              $\Si_{red}$. This is realized by a portion of $\{ v=\text{const}\}$ in the ingoing Eddington-Finkelstein coordinates.
             \item   Strictly spacelike  in $\{\frac{11}{10} r_{\mathcal{H}} < r < R \}$ with  $R$  a fixed  number  $R \gg r_P$. We denote  this  portion   by $\Si_{trap}$. This is realized by a portion of $\widetilde{\Sigma}_\tau=\{ t=\tau \}$. 
                          \item     Outgoing null    in  $\{  r  \ge R \}$
             with $e_4$ as null  outgoing generator.   We denote this portion by $\Si_{far}$. This is realized by a portion of $\{ u= \text{const}\}$ in the outgoing Eddington-Finkelstein coordinates. 
                      \end{enumerate}
 We denote $\MM(\tau_1, \tau_2)\subset \MM$  the spacetime region   in  the past of $\Si(\tau_2)$ and in the future of $\Si(\tau_1)$.  For fixed $R$ we denote by $\MM_{\le R}$ and  $\MM_{\ge R} $ the   regions  defined by $r\le R$ and $r\ge R$. We denote by $\Sigma_{ \geq R} (\tau)$ the portion of hypersurface for $r \geq R$. See Figure 1.

 Let $p$ be a free parameter with $\de \leq p \leq 2-\de$, for $\de>0$, as in standard application of the $r^p$-method of Dafermos-Rodnianski \cite{rp}.

    We introduce the following  weighted energies for $\Phi_1$ and $\Phi_2$. 
    \begin{enumerate}
    \item Energy quantities on $\Si_\tau$:
    \begin{itemize}
    \item Basic energy quantity
    \bea
    \begin{split}
\label{def:basic-energy-deg}
E[\Phi_1, \Phi_2](\tau):&= \int_{\Si_{red}}    |\nabb^*_3(\Phi_1)|^2 +|\nabb\Phi_1|^2 + |\Phi_1|^2 +   |\nabb^*_3(\Phi_2)|^2 +|\nabb\Phi_2|^2 + |\Phi_2|^2 \\
&+ \int_{\Si_{trap}}    |\nabb_4 (\Phi_1)|^2  + |\nabb_3(\Phi_1)|^2 +|\nabb\Phi_1|^2 + r^{-2}|\Phi_1|^2 \\
&+\int_{\Si_{trap}}    |\nabb_4 (\Phi_2)|^2  + |\nabb_3(\Phi_2)|^2 +|\nabb\Phi_2|^2 + r^{-2}|\Phi_2|^2 \\
&+ \int_{\Si_{far}}    |\nabb_4 (\Phi_1)|^2   +|\nabb\Phi_1|^2 + r^{-2}|\Phi_1|^2 + |\nabb_4 (\Phi_2)|^2  +|\nabb\Phi_2|^2 + r^{-2}|\Phi_2|^2
\end{split}
\eea
Notice that the above energy quantity is regular up to the horizon. Observe that along outgoing null hypersurfaces $\{u=\text{const}\}$ the volume form is given by $dv d\text{vol}_{S_{u,v}}$, and along ingoing null hypersurfaces $\{ v=\text{const} \}$ the volume form is $du d\text{vol}_{S_{u,v}}$, where $d\text{vol}_{S_{u,v}}$ denotes the volume form of the sphere $S_{u, v}$. 
\item Weighted energy quantity in the far away region
\beaa
  E_{p\,; \,R}[\Phi_1, \Phi_2](\tau)&:= & \int_{\Si_{r \ge  R}(\tau)}  r^p |\check{\nabb}_4\Phi_1|^2+ r^p |\check{\nabb}_4\Phi_2|^2
  \eeaa
  where $\check{\nabb}_4\Psi:=\nabb_4\Psi+\frac 1 r \Psi$. 
  \item Weighted energy quantity
  \beaa
  E_{p}[\Phi_1, \Phi_2](\tau)&:=&E[\Phi_1, \Phi_2](\tau)+ E_{p\,; \,R}[\Phi_1, \Phi_2](\tau)
  \eeaa
\end{itemize}

    \item  Weighted spacetime bulk energies in $\MM(\tau_1, \tau_2)$:
    \begin{itemize}
    \item Basic Morawetz bulk 
    \bea
    \label{def:Mor-bulk}
    \begin{split}
&\Mor[\Phi_1, \Phi_2](\tau_1, \tau_2)\\
&:= \int_{\MM(\tau_1, \tau_2)} \frac{1}{r^3} |R(\Phi_1)|^2+ \frac{1}{r^4} |\Phi_1|^2 + \frac{(r^2-3M r +2Q^2)^2}{r^5}\left( |\nabb \Phi_1|^2+\frac{1}{r^2} |T\Phi_1|^2 \right)\\
&+ \int_{\MM(\tau_1, \tau_2)} \frac{1}{r^3} |R(\Phi_2)|^2+ \frac{1}{r^4} |\Phi_2|^2 + \frac{(r^2-3M r +2Q^2)^2}{r^5}\left( |\nabb \Phi_2|^2+\frac{1}{r^2} |T\Phi_2|^2 \right)   
\end{split}  
\eea
where $R=\frac 1 2(-\Up e_3^*+e_4^*)=\frac 1 2 (-e_3+\Up e_4)$. 
Notice that the Morawetz bulk $\Mor[\Psi](\tau_1, \tau_2)$ is degenerate at the photon sphere $\{ r=r_P\}$. 
\item  Weighted bulk norm in the far away region 
\bea\label{Morawetz-far-away}
\begin{split}
\MM_{p\,; \,R}[\Phi_1, \Phi_2](\tau_1, \tau_2):&=\int_{\MM_{r\ge  R}(\tau_1, \tau_2) }  r^{p-1}  \left( p | \check{\nabb}_4(\Phi_1) |^2 +(2-p)  ( |\nabb \Phi_1|^2+   r^{-2}  |\Phi_1|^2)\right)  \\
&+\int_{\MM_{r\ge  R}(\tau_1, \tau_2) }  r^{p-1}  \left( p | \check{\nabb}_4(\Phi_2) |^2 +(2-p)  ( |\nabb \Phi_2|^2+   r^{-2}  |\Phi_2|^2)\right)  
\end{split}
\eea
Observe that for $\de \leq p \leq 2-\de$, for $\de>0$, the above bulk norm is positive definite. 
\item Weighted bulk norm
\beaa
\MM_{p}[\Phi_1, \Phi_2](\tau_1, \tau_2)&:=&\Mor[\Phi_1, \Phi_2](\tau_1, \tau_2)+\MM_{p\,; \,R}[\Phi_1, \Phi_2](\tau_1, \tau_2)
\eeaa
\end{itemize}
    \end{enumerate}

We can now state the main theorem in terms of the above energy quantities.

\begin{theorem}\label{main-theorem} Let $\Phi_1$ and $\Phi_2$ be a 1-tensor and a symmetric traceless 2-tensor respectively, satisfying equations \eqref{box-phi-1} and \eqref{box-phi-2} in Reissner-Nordstrom spacetime with $|Q|<M$ and supported on the $\ell \geq2$ spherical harmonics. Then, 
\begin{enumerate}
\item  for all $\de \leq p \leq 2-\de$ and for any $\tau>0$, boundedness of the weighted energy holds:
\bea\label{main-1}
E_{p}[\Phi_1, \Phi_2](\tau) \leq C   E_{p}[\Phi_1, \Phi_2](0)
\eea
\item for all $\de \leq p \leq 2-\de$ and for any $\tau>0$, the following integrated local energy decay estimates for $\Phi_1$ and $\Phi_2$ holds:
\bea\label{main-2}
\MM_{p}[\Phi_1, \Phi_2](0, \tau) \leq C   E_{p}[\Phi_1, \Phi_2](0)
\eea
\end{enumerate}
\end{theorem}

Observe that, as in the case of integrated local energy decay estimates for even the scalar linear wave equation on black hole backgrounds, the degeneracy at the photon sphere of the Morawetz bulk cannot be eliminated. The degeneracy is caused by the presence of orbital trapped geodesics, which are an obstruction to decay \cite{lectures}.

In the following we show how to obtain the boundedness of the energy in Section \ref{sec-en}, the Morawetz estimates in Section \ref{sec-mo} and the $r^p$-estimates in Section \ref{rp-sec}. Combining those estimates, we prove the boundedness of the the weighted energy \eqref{main-1} and the integrated local energy decay estimates \eqref{main-2}, therefore proving Theorem \ref{main-theorem}.

\section{The energy-momentum tensor associated to the system}\label{str-sy}

In this section, we define a combined energy-momentum tensor associated to the system formed by equations \eqref{box-phi-1} and \eqref{box-phi-2}.  We first recall the definition of energy-momentum tensor and current associated to a solution of a wave equation. 

We define the stress-energy tensor of equations \eqref{box-phi-1} and \eqref{box-phi-2} as:
\beaa
\TT_{\mu\nu}[\Phi_1]:&=& D_\mu \Phi_1 \c D_\nu \Phi_1-\frac 1 2 g_{\mu \nu}\left( D_{\lambda} \Phi_1 \c D^{\lambda} \Phi_1 +V_1 |\Phi_1|^2\right)=D_\mu \Phi_1 \c D_\nu \Phi_1-\frac 1 2 g_{\mu \nu} \LL_1[\Phi_1]\\
\TT_{\mu\nu}[\Phi_2]:&=& D_\mu \Phi_2 \c D_\nu \Phi_2-\frac 1 2 g_{\mu \nu}\left( D_{\lambda} \Phi_2 \c D^{\lambda} \Phi_2 +V_2 |\Phi_2|^2\right)= D_\mu \Phi_2 \c D_\nu \Phi_2-\frac 1 2 g_{\mu \nu}\LL_2[\Phi_2]
\eeaa
where $D$ denotes the covariant derivative of the Reissner-Nordstr\"om metric $g=g_{M, Q}$, and $V_1$ and $V_2$ are the potentials of the equations defined in \eqref{potential-1} and \eqref{potential-2}.

For  $X$ a vectorfield, $w$ a scalar function and $M$ a one form, we define the associated currents as: 
\beaa
 \PP_\mu^{(X, w_1, M_1)}[\Phi_1]&:=&\TT_{\mu\nu}[\Phi_1] X^\nu +\frac 1 2  w_1 \Phi_1 D_\mu \Phi_1 -\frac 1 4   \pr_\mu w_1  |\Phi_1|^2+\frac 1 4  M_{1\mu} |\Phi_1|^2\\
  \PP_\mu^{(X, w_2, M_2)}[\Phi_2]&:=&\TT_{\mu\nu}[\Phi_2] X^\nu +\frac 1 2  w_2 \Phi_2 D_\mu \Phi_2 -\frac 1 4   \pr_\mu w_2  |\Phi_2|^2+\frac 1 4  M_{2\mu} |\Phi_2|^2
  \eeaa
From a standard computation (see for example \cite{Giorgi4}) we obtain that for $X=a(r)e_3+b(r)e_4$, in a spherically symmetric spacetime, the divergence of the current is given by
  \bea
  \label{le:divergPP-gen-new}
  \begin{split}
  D^\mu \PP_\mu^{(X, w_1, M_1)}[\Phi_1]&= \frac 1 2 \TT[\Phi_1]  \c\piX+\left( - \frac 1 2 X( V_1 ) -\frac 1 4   \Box_g  w_1 \right)|\Phi_1|^2+\frac 12  w_1 \LL_1[\Phi_1] \\
  &+\frac 1 4  D^\mu (\Phi_1^2 M_\mu)    +  \left(X(\Phi_1)+\frac 1 2   w_1 \Phi_1\right)\c \frac{8Q}{r^2}\DDd_2\Phi_2
   \end{split}
 \eea
 and 
\bea
  \label{le:divergPP-gen-2-new}
  \begin{split}
  D^\mu \PP_\mu^{(X, w_2, M_2)}[\Phi_2]&= \frac 1 2 \TT[\Phi_2]  \c\piX+\left( - \frac 1 2 X( V_2 ) -\frac 1 4   \Box_g  w_2 \right)|\Phi_2|^2+\frac 12  w_2 \LL_2[\Phi_2]\\
  & +\frac 1 4  D^\mu (|\Phi_2|^2 M_\mu)    +  \left(X(\Phi_2)+\frac 1 2   w_2 \Phi_2\right)\c \frac{Q}{r^2}\DDs_2\Phi_1 
   \end{split}
 \eea
 where $\piX$ is the deformation tensor of the vectorfield $X$. 
 
 Notice the presence of the right hand sides of the equations, $\frac{8Q}{r^2}\DDd_2\Phi_2$ and $\frac{Q}{r^2}\DDs_2\Phi_1$ in \eqref{le:divergPP-gen-new} and \eqref{le:divergPP-gen-2-new}. In particular, those right hand sides appear in the divergence of the current. Our goal is to define an energy-momentum tensor for the system which has good cancellation properties with respect to these additional terms. It turns out that the system composed by equations \eqref{box-phi-1} and \eqref{box-phi-2} admits a conserved energy-momentum tensor.

 \begin{definition}\label{def-str} Let $\Phi_1$ and $\Phi_2$ be a 1-tensor and a symmetric traceless 2-tensor respectively, satisfying the system of coupled wave equations \eqref{box-phi-1} and \eqref{box-phi-2}. We define the energy-momentum tensor for the system as the following symmetric two tensor $\TT_{\mu\nu}[\Phi_1, \Phi_2]$: 
\bea\label{definition-combined-stress-energy-tensor}
\TT_{\mu\nu}[\Phi_1, \Phi_2]&:=& \TT_{\mu\nu}[\Phi_1]+8\TT_{\mu\nu}[\Phi_2]-\frac{8Q}{r^2} \left(\DDs_2 \Phi_1 \c \Phi_2 \right) g_{\mu\nu}
\eea
where $\TT_{\mu\nu}[\Phi_1]$ and $\TT_{\mu\nu}[\Phi_2]$ are the standard energy-momentum tensors associated to equations \eqref{box-phi-1} and \eqref{box-phi-2} respectively. 

We also define the associated combined current $\PP^{(X, w, M)}_\mu[\Phi_1, \Phi_2]$ for a vectorfield $X$, a pair of scalar functions $w=(w_1, w_2)$, a pair of one forms $M=(M_1, M_2)$ as
\bea\label{definition-combined-current}
\PP^{(X, w, M)}_\mu[\Phi_1, \Phi_2]&=&  \PP_\mu^{(X, w_1, M_1)}[\Phi_1]+8  \PP_\mu^{(X, w_2, M_2)}[\Phi_2]-\frac{8Q}{r^2} \left(\DDs_2 \Phi_1 \c \Phi_2 \right) X_\mu
\eea
\end{definition}

The above definition is motivated by the cancellation properties for the divergence of the new combined current $\PP^{(X, w, M)}_\mu[\Phi_1, \Phi_2]$, as showed in the following lemma. 

\begin{lemma}\label{divergence-combined-current} Let $\Phi_1$ and $\Phi_2$ be a 1-tensor and a symmetric traceless 2-tensor respectively, satisfying the system of coupled wave equations \eqref{box-phi-1} and \eqref{box-phi-2}. For $X=a(r)e_3+b(r)e_4$, the divergence of the combined current $\PP^{(X, w, M)}_\mu[\Phi_1, \Phi_2]$ is given by
\bea\label{divergence-P-EE}
D^\mu \PP^{(X, w, M)}_\mu[\Phi_1, \Phi_2]&=_s& \EE^{(X, w, M)}[\Phi_1, \Phi_2]
\eea
where $=_s$ indicates that the equality holds upon integration on the sphere and
\bea\label{definition-EE}
\begin{split}
&\EE^{(X, w, M)}[\Phi_1, \Phi_2]\\
&:=  \frac 1 2 \TT[\Phi_1]  \c\piX+\left( - \frac 1 2 X( V_1 ) -\frac 1 4   \Box_g  w_1 \right)|\Phi_1|^2+\frac 12  w_1 \LL_1[\Phi_1] +\frac 1 4  D^\mu (|\Phi_1|^2 M_\mu)      \\
 &+4 \TT[\Phi_2]  \c\piX+\left( - 4 X( V_2 ) -2   \Box_g  w_2 \right)|\Phi_2|^2+4  w_2 \LL_2[\Phi_2] +2 D^\mu (|\Phi_2|^2 M_\mu)      \\
  &+\frac{ 4Q}{r^2}\left(w_1+w_2 +\frac{4}{r}  X(r) - \tr \piX\right)\DDs_2\Phi_1 \c  \Phi_2  -\frac{8Q}{r^2}([X, \DDs_2] \Phi_1) \c \Phi_2  
  \end{split}
\eea
\end{lemma}
\begin{proof} We compute, using \eqref{le:divergPP-gen-new} and \eqref{le:divergPP-gen-2-new}:
\beaa
D^\mu\PP^{(X, w, M)}_\mu[\Phi_1, \Phi_2]&=&  \frac 1 2 \TT[\Phi_1]  \c\piX+\left( - \frac 1 2 X( V_1 ) -\frac 1 4   \Box_g  w_1 \right)|\Phi_1|^2+\frac 12  w_1 \LL_1[\Phi_1] +\frac 1 4  D^\mu (\Phi_1^2 M_\mu)      \\
  &&+  \left(X(\Phi_1)+\frac 1 2   w_1 \Phi_1\right)\c \frac{8Q}{r^2}\DDd_2\Phi_2\\
  &&+4 \TT[\Phi_2]  \c\piX+\left( - 4 X( V_2 ) -2   \Box_g  w_2 \right)|\Phi_2|^2+4  w_2 \LL_2[\Phi_2] +2 D^\mu (|\Phi_2|^2 M_\mu)      \\
  &&+  \left(X(\Phi_2)+\frac 1 2   w_2 \Phi_2\right)\c \frac{8Q}{r^2}\DDs_2\Phi_1 \\
&&-X(\frac{8Q}{r^2}) \left(\DDs_2 \Phi_1 \c \Phi_2 \right) -\frac{8Q}{r^2} \left(X(\DDs_2 \Phi_1) \c \Phi_2 \right)-\frac{8Q}{r^2} \left(\DDs_2 \Phi_1 \c X(\Phi_2) \right)\\
&& -\frac{8Q}{r^2} \left(\DDs_2 \Phi_1 \c \Phi_2 \right) D^\mu(X_\mu)
\eeaa
where the last two lines are the divergence of the additional term $-\frac{8Q}{r^2} \left(\DDs_2 \Phi_1 \c \Phi_2 \right) X_\mu$ in the definition of $\PP^{(X, w, M)}_\mu[\Phi_1, \Phi_2]$. 
Recall that $D^\mu(X_\mu)=\frac 1 2 \tr \piX$. We compute
\beaa
-\frac{8Q}{r^2} \left(X(\DDs_2 \Phi_1) \c \Phi_2 \right)&=&-\frac{8Q}{r^2} \Big(\DDs_2 (X(\Phi_1)) \c \Phi_2+([X, \DDs_2] \Phi_1) \c \Phi_2 \Big)\\
&=_s&  -\frac{8Q}{r^2} \Big( X(\Phi_1) \c \DDd_2\Phi_2+([X, \DDs_2] \Phi_1) \c \Phi_2 \Big)
\eeaa
where the last equality holds upon integration on the spheres, where we used that $\DDd_2$ and $\DDs_2$ are adjoint operators on the sphere as in \eqref{adjoint}.  We then obtain the cancellation of the terms $X(\Phi_2) \c \frac{8Q}{r^2} \DDs_2 \Phi_1$ and $X(\Phi_1) \c \frac{8Q}{r^2} \DDs_2 \Phi_2$. This implies the lemma.
\end{proof}

\section{Boundedness of the energy}\label{sec-en}

In this section we prove boundedness of the energy for the mixed spin $\pm1$ and spin $\pm2$ system of equations. 

To derive the energy estimates we apply Lemma \ref{divergence-combined-current} to the Killing vectorfield $X=T$, with $w=0$, $M=0$. We obtain 
\bea\begin{split}\label{energy-}
D^\mu \PP^{(T, 0, 0)}_\mu[\Phi_1, \Phi_2]&=_s  \frac 1 2 \TT[\Phi_1]  \c\piT - \frac 1 2 T( V_1 ) |\Phi_1|^2  +4 \TT[\Phi_2]  \c\piT - 4 T( V_2 )|\Phi_2|^2    \\
  &+\frac{ 4Q}{r^2}\left(\frac{4}{r}  T(r) - \tr \piT\right)\DDs_2\Phi_1 \c  \Phi_2  -\frac{8Q}{r^2}([T, \DDs_2] \Phi_1) \c \Phi_2\\
  &=_s 0
  \end{split}
\eea
since $\piT=0$,  $T(r)=T(V_1)=T(V_2)=0$ and $T$ commutes with the angular operator. 

By applying the divergence theorem to $D^\mu \PP^{(T, 0, 0)}_\mu[\Phi_1, \Phi_2]=_s 0 $ in the region $\MM(0, \tau)$, we are left to analyze the boundary terms $\int_{\Sigma_\tau} \PP^{(T, 0, 0)}_\mu[\Phi_1, \Phi_2] \c n_{\Sigma_\tau}$,
where $n_{\Sigma_\tau}$ is the normal vector to $\Sigma_\tau$. Explicitly, $n_{\Sigma_{\tau}}=e_3^*$ in $\Sigma_{red}$, $n_{\Sigma_{\tau}}=\frac{1}{\Up} T$ in $\Sigma_{trap}$, $n_{\Sigma_{\tau}}=e_4$ in $\Sigma_{far}$ and along $\mathcal{H}^+$, and $n_{\Sigma_{\tau}}=e_3$ along $\mathscr{I}^+$. In particular, $g(T, n_{\Sigma_{\tau}})=-1$.
Our goal is to show that 
 $\int_{\Sigma_\tau} \PP^{(T, 0, 0)}_\mu[\Phi_1, \Phi_2] \c n_{\Sigma_\tau}$ is positive definite, and comparable with $\TT_{\mu\nu}[\Phi_1]T^\mu n_{\Sigma_\tau}+8\TT_{\mu\nu}[\Phi_2]T^\mu n_{\Sigma_\tau}$ (and therefore with $E[\Phi_1, \Phi_2](\tau)$). 
 Observe that $V_1=\frac{1}{r^2}\Up+\frac{5Q^2}{r^4} \geq 0$, and $V_2=\frac{4}{r^2}\Up+\frac{2Q^2}{r^4} \geq 0$ are positive in the whole exterior region, therefore $\TT_{\mu\nu}[\Phi_1]T^\mu n_{\Sigma_t}+8\TT_{\mu\nu}[\Phi_2]T^\mu n_{\Sigma_t}$ is coercive.

We compute 
\beaa
\PP^{(T, 0, 0)}_\mu[\Phi_1, \Phi_2] \c n_{\Sigma_\tau}&=& \TT_{\mu\nu}[\Phi_1]T^\mu n^\nu_{\Sigma_\tau}+8\TT_{\mu\nu}[\Phi_2]T^\mu n^\nu_{\Sigma_\tau}-\frac{8Q}{r^2}\left( \DDs_2 \Phi_1 \c \Phi_2\right)\g(T,n_{\Sigma_\tau})
\eeaa
In $\Sigma_{red}$, where $T=\frac 1 2 (\Up e_3^*+e_4^*)$, the above boundary term becomes
\beaa
\PP^{(T, 0, 0)}_\mu[\Phi_1, \Phi_2] \c n_{\Sigma_\tau}&=& \frac 1 2 \Up \TT(e^*_3, e^*_3)[\Phi_1]+\frac 1 2 \TT(e^*_3, e^*_4)[\Phi_1]+4\Up \TT(e^*_3, e^*_3)[\Phi_2]+4 \TT(e^*_3, e^*_4)[\Phi_2]\\
&&+\frac{8Q}{r^2}\left( \DDs_2 \Phi_1 \c \Phi_2\right)\\
&=& \frac{ \Up}{ 2} |\nabb_3^*\Phi_1|^2+\frac 1 2|\nabb \Phi_1|^2 +\frac 1 2 V_1|\Phi_1|^2 +4\Up |\nabb_3^*\Phi_2|^2+4|\nabb \Phi_2|^2 +4 V_2|\Phi_2|^2\\
&&+\frac{8Q}{r^2} \DDs_2 \Phi_1 \c \Phi_2
\eeaa
In $\Sigma_{trap}$, where $T=\frac 1 2 (e_3+\Up e_4)$, the boundary term becomes
\beaa
\PP^{(T, 0, 0)}_\mu[\Phi_1, \Phi_2] \c n_{\Sigma_\tau}&=& \frac{ 1}{ 4 \Up} \TT_{33}[\Phi_1]+\frac 1 2 \TT_{34}[\Phi_1]+\frac 1 4 \Up^2 T_{44}[\Phi_1]+ \frac{ 2}{ \Up} \TT_{33}[\Phi_2]+4 \TT_{34}[\Phi_2]\\
&&+2 \Up^2 \TT_{44}[\Phi_2]+\frac{8Q}{r^2}\left( \DDs_2 \Phi_1 \c \Phi_2\right)\\
&=&\frac{ 1}{ 4 \Up} |\nabb_3\Phi_1|^2+\frac 1 4 \Up^2 |\nabb_4\Phi_1|^2+\frac 1 2|\nabb \Phi_1|^2 +\frac 1 2 V_1|\Phi_1|^2 \\
&&+\frac{ 2}{ \Up}  |\nabb_3\Phi_2|^2+2\Up^2|\nabb_4\Phi_2|^2+4|\nabb \Phi_2|^2 +4 V_2|\Phi_2|^2+\frac{8Q}{r^2} \DDs_2 \Phi_1 \c \Phi_2
\eeaa
Similarly, in $\Sigma_{far}$ we have 
\beaa
\PP^{(T, 0, 0)}_\mu[\Phi_1, \Phi_2] \c n_{\Sigma_\tau}&=&\frac 1 2 \Up T_{44}[\Phi_1]+\frac 1 2 \TT_{34}[\Phi_1]+4 \Up T_{44}[\Phi_2]+4 \TT_{34}[\Phi_2]+\frac{8Q}{r^2}\left( \DDs_2 \Phi_1 \c \Phi_2\right)\\
&=&\frac 1 2 \Up |\nabb_4\Phi_1|^2+\frac 1 2|\nabb \Phi_1|^2 +\frac 1 2 V_1|\Phi_1|^2+4\Up|\nabb_4\Phi_2|^2+4|\nabb \Phi_2|^2 +4 V_2|\Phi_2|^2\\
&&+\frac{8Q}{r^2} \DDs_2 \Phi_1 \c \Phi_2
\eeaa
The boundary terms at the event horizon and at future null infinity can be similarly analyzed. In particular, in each portion of $\Sigma_\tau$ we can estimate the boundary terms by
\beaa
\PP^{(T, 0, 0)}_\mu[\Phi_1, \Phi_2] \c n_{\Sigma_\tau}&\geq &\frac 1 2|\nabb \Phi_1|^2 +\frac 1 2 V_1|\Phi_1|^2+4|\nabb \Phi_2|^2 +4 V_2|\Phi_2|^2+\frac{8Q}{r^2} \DDs_2 \Phi_1 \c \Phi_2\\
&\geq &\frac 1 2|\nabb \Phi_1|^2 +\frac{5Q^2}{2r^4}|\Phi_1|^2+4|\nabb \Phi_2|^2 +\frac{8Q^2}{r^4}|\Phi_2|^2+\frac{8Q}{r^2} \DDs_2 \Phi_1 \c \Phi_2
\eeaa
since $V_1=\frac{1}{r^2}\Up+\frac{5Q^2}{r^4} \geq \frac{5Q^2}{r^4}$, and $V_2=\frac{4}{r^2}\Up +\frac{2Q^2}{r^4}\geq \frac{2Q^2}{r^4}$. 

We now show that the above right hand side defines a positive definite quadratic form, therefore implying that $\PP^{(T, 0, 0)}_\mu[\Phi_1, \Phi_2] \c n_{\Sigma_t}$ is positive definite.

Suppose that $\Phi_1$ and $\Phi_2$ are supported on the fixed $\ell \geq 2$ spherical harmonic. Using  \eqref{elliptic-spj}, \eqref{elliptic-spj-2} and \eqref{dot-product-estimate-l} we bound 
\beaa
&& \frac 1 2|\nabb \Phi_1|^2  +\frac{5Q^2}{2r^4}|\Phi_1|^2 +4|\nabb \Phi_2|^2 +\frac{8Q^2}{r^4}|\Phi_2|^2+\frac{8Q}{r^2} \DDs_2 \Phi_1 \c \Phi_2 \\
 &\geq &\frac{1}{ 2r^2}\left(\lambda-1+\frac{5Q^2}{r^2}\right) |\Phi_1|^2   +\frac{4}{r^2} \left(\lambda-4+\frac{2Q^2}{r^2}\right) |\Phi_2|^2 -\frac{4Q}{r^2}  \frac{(2\lambda-4)^{1/2}}{r} |\Phi_1 | |\Phi_2|
\eeaa
where we denoted $\lambda:=\ell(\ell+1) \geq 6$.

The above is a quadratic form of the type $a|\Phi_1|^2 +b |\Phi_2|^2-c|\Phi_1||\Phi_2|$. Its discriminant is given by $D=c^2-4ab$, and if the discriminant is negative, then the quadratic form is positive definite.
We compute the discriminant of the above quadratic form: 
\beaa
-D &=&  \frac{8}{ r^4}\Big[\left(\lambda-1+\frac{5Q^2}{r^2}\right) \left(\lambda-4+\frac{2Q^2}{r^2}\right)-\frac{4Q^2}{r^2}  (\lambda-2)\Big]\\
&\geq & \frac{8}{ r^4}\Big[\left(\lambda-1\right) \left(\lambda-4\right)+\frac{2Q^2}{r^2}\left(\lambda-1\right) +\frac{5Q^2}{r^2} \left(\lambda-4\right)-\frac{4Q^2}{r^2}  (\lambda-2)\Big]\\
&= & \frac{8}{ r^4}\Big[\left(\lambda-1\right) \left(\lambda-4\right)+\frac{Q^2}{r^2} \left(3\lambda-14\right)\Big]
\eeaa
Since $\lambda \geq 6$, we have that $3\lambda-14>0$, therefore implying positivity of the above quadratic form. 

To obtain the estimates for the non-degenerate energy, we make use of the celebrated redshift vectorfield, as introduced in \cite{redshift}. Notice that the non-degeneracy along the horizon given by the redshift vectorfield fails exactly at the extremal case $|Q|=M$. We have  for $X=a(r) e_3+b(r) e_4$ (see Lemma 5 in \cite{Giorgi4})
 \beaa
\TT\c \piX  &=& \left( \Up a'- b'+\left(\frac{2M}{r^2}-\frac{2Q^2}{r^3} \right)a \right)|\nabb \Psi|^2+ \left( \Up b'+\left(-\frac{2M}{r^2}+\frac{2Q^2}{r^3} \right)b \right)|\nabb_4 \Psi|^2 -a'|\nabb_3 \Psi|^2\\
 &&+\left(-\frac{2\Up}{r}a+\frac{2}{r}b\right)\nabb_3 \Psi \c \nabb_4 \Psi+\left(\Up a'- b'+\left(\frac{2M}{r^2}-\frac{2Q^2}{r^3} +\frac{2\Up}{r}\right)a -\frac{2}{r}b\right)V_i |\Psi|^2
 \eeaa
  Consider a vector field $N$ defined as $N= a(r) \nabb_3+b(r) \nabb_4$, and such that the functions $a(r)$ and $b(r)$ verify
\beaa
a(r_\HH)=0, \qquad   b(r_\HH)=-1.
\eeaa
Clearly one can choose the functions $a$ and $b$ such that $\EE^{(N, 0, 0)}[\Phi_1, \Phi_2]$ is positive definite. Notice that the coefficient of $\nabb_4$ along the horizon reduces to $\frac{2M}{r_{\mathcal{H}}^2}-\frac{2Q^2}{r_{\mathcal{H}}^3} $, which degenerates to zero at the horizon of an extremal Reissner-Nordstr\"om, for which $r_{\mathcal{H}}=M=Q$. Therefore the above bulk fails to be positive definite in the extremal case.

We have therefore obtained the following. 
\begin{proposition}\label{prop-energy-deg}[Boundedness of the energy] Let $\Phi_1$ and $\Phi_2$ be a 1-tensor and a symmetric traceless 2-tensor respectively, satisfying the system of coupled wave equations \eqref{box-phi-1} and \eqref{box-phi-2} in Reissner-Nordstr\"om spacetime with $|Q| < M$ and supported in $\ell \geq 2$ spherical harmonics. Then we have
\beaa
E[\Phi_1, \Phi_2](\tau_2) \leq C E[\Phi_1, \Phi_2](\tau_1)
\eeaa
for every $\tau_1 \leq \tau_2$. 
\end{proposition}

\section{Morawetz estimates}\label{sec-mo}

In this section we prove Morawetz estimates for the mixed spin $\pm1$ and spin $\pm2$ system of equations.
 
To derive the Morawetz estimates, we apply Lemma \ref{divergence-combined-current}  to the radial vector field $Y= f(r) R$, for $R=\Up \partial_r$, and a function $f(r)$ to be determined.

 \begin{proposition}
  \label{prop-ident-Mor1}
 Let  $Y=f(r) R$  and   $w=  r^{-2}\Up\pr_r \left( r^2 f(r)\right)  $. Let $\Phi_1$ and $\Phi_2$ be a 1-tensor and a symmetric traceless 2-tensor respectively, satisfying the system of coupled wave equations \eqref{box-phi-1} and \eqref{box-phi-2}. Then we have 
  \beaa
\EE^{(Y, w, 0)}[\Phi_1, \Phi_2]&=& \frac{ f}{r}\left( 1-\frac{3M}{r}+\frac{2Q^2}{r^2}  \right)|\nabb \Phi_1|^2+ f'|R \Phi_1|^2+\left(-\frac 1 2 \pr_r\left(\Up V_1\right)f -\frac 1 4   \Box_g  w \right)|\Phi_1|^2   \\
 &&+\frac{8f}{r}\left( 1-\frac{3M}{r}+\frac{2Q^2}{r^2}  \right)|\nabb \Phi_2|^2+8 f'|R \Phi_2|^2+8\left(-\frac 1 2 \pr_r\left(\Up V_2\right)f -\frac 1 4    \Box_g  w \right)|\Phi_2|^2    \\
  &&+\frac{ 16 Q}{r^3} f\left(1-\frac{3M}{r}+\frac{2Q^2}{r^2} +\frac 1 2 \Up\right)\DDs_2\Phi_1 \c  \Phi_2
\eeaa
  \end{proposition}
 \begin{proof}
By \eqref{definition-EE}, we have
\beaa
\EE^{(Y, w, 0)}[\Phi_1, \Phi_2]&=&  \frac 1 2 \TT[\Phi_1]  \c\piY+\left( - \frac 1 2 Y( V_1 ) -\frac 1 4   \Box_\g  w_1 \right)|\Phi_1|^2+\frac 12  w_1 \LL_1[\Phi_1]     \\
 &&+4 \TT[\Phi_2]  \c\piY+\left( - 4 Y( V_2 ) -2   \Box_\g  w_2 \right)|\Phi_2|^2+4  w_2 \LL_2[\Phi_2]      \\
  &&+\frac{ 4Q}{r^2}\left(w_1+w_2 +\frac{4}{r}  Y(r) - \tr \piY\right)\DDs_2\Phi_1 \c  \Phi_2  -\frac{8Q}{r^2}([Y, \DDs_2] \Phi_1) \c \Phi_2  
\eeaa
We have for $Y=f(r) R$ (see Corollary 4 in \cite{Giorgi4})
\beaa
 \TT[\Phi_1] \c \piY&=& 2f\left( \frac{1}{r}-\frac{3M}{r^2}+\frac{2Q^2}{r^3}  \right)|\nabb \Phi_1|^2+2 f'|R \Phi_1|^2+\left(-\frac{2\Up}{r}f-\Up f'\right) \LL_1[\Phi_1]\\
 &&+\left(-\frac{2M}{r^2}+\frac{2Q^2}{r^3}\right) fV_1 |\Phi_1|^2 
 \eeaa
 and similarly for $\Phi_2$. We therefore obtain
\beaa
\EE^{(Y, w)}[\Phi_1, \Phi_2]&=&  f\left( \frac{1}{r}-\frac{3M}{r^2}+\frac{2Q^2}{r^3}  \right)|\nabb \Phi_1|^2+ f'|R \Phi_1|^2+\left(\frac 12  w_1 -\frac{\Up}{r}f-\frac{\Up}{2} f'\right) \LL_1[\Phi_1]\\
 &&+\left( - \frac 1 2 Y( V_1 )+\left(-\frac{M}{r^2}+\frac{Q^2}{r^3}\right) fV_1  -\frac 1 4   \Box_g  w_1 \right)|\Phi_1|^2   \\
 &&+8f\left( \frac{1}{r}-\frac{3M}{r^2}+\frac{2Q^2}{r^3}  \right)|\nabb \Phi_2|^2+8 f'|R \Phi_2|^2+8\left(\frac 12  w_2 -\frac{\Up}{r}f-\frac{\Up}{2} f'\right) \LL_2[\Phi_2]\\
 &&+8\left( - \frac 1 2  Y( V_2 )+\left(-\frac{M}{r^2}+\frac{Q^2}{r^3}\right) fV_2 -\frac 1 4    \Box_g  w_2 \right)|\Phi_2|^2    \\
  &&+\frac{ 4Q}{r^2}\left(w_1+w_2 +\frac{4}{r}  Y(r) - \tr \piY\right)\DDs_2\Phi_1 \c  \Phi_2  -\frac{8Q}{r^2}([Y, \DDs_2] \Phi_1) \c \Phi_2  
\eeaa
With the choice 
$w:=w_1=w_2=  r^{-2}\Up\pr_r \left( r^2 f(r)\right)= \frac{2\Up}{r}f+\Up f'$ the terms involving $\LL_1[\Phi_1]$ and $\LL_2[\Phi_2]$ above cancel out. 
 Using that (see Corollary 3 in \cite{Giorgi4})
 \beaa
  \tr\piY&=& \frac{4\Up}{r} f +\left(\frac{4M}{r^2}-\frac{4Q^2}{r^3} \right)f+2f' \Up =\frac{4}{r} \left( 1-\frac{M}{r} \right) f + 2 f' \Up
 \eeaa
we compute, recalling that $R(r)=\Up$, 
 \beaa
  [Y, \DDs_2]&=& f[ R, \DDs_2]=-\frac{\Up}{r}f \DDs_2\\
 w_1+w_2 +\frac{4}{r}  Y(r) - \tr \piY&=& 2(\frac{2\Up}{r}f+\Up f')+\frac{4}{r}  \Up f - (\frac{4\Up}{r} f +\left(\frac{4M}{r^2}-\frac{4Q^2}{r^3} \right)f+2f' \Up)\\
 &=& \frac{4\Up}{r}f - \left(\frac{4M}{r^2}-\frac{4Q^2}{r^3} \right)f= \frac 4 r  \left(1-\frac{3M}{r}+\frac{2Q^2}{r^2} \right)f
 \eeaa
 We conclude
 \beaa
&&\EE^{(Y, w, 0)}[\Phi_1, \Phi_2]\\
&=&  f\left( \frac{1}{r}-\frac{3M}{r^2}+\frac{2Q^2}{r^3}  \right)|\nabb \Phi_1|^2+ f'|R \Phi_1|^2+8f\left( \frac{1}{r}-\frac{3M}{r^2}+\frac{2Q^2}{r^3}  \right)|\nabb \Phi_2|^2+8 f'|R \Phi_2|^2\\
&&+\left(\left( - \frac \Up 2 V_1'+\left(-\frac{M}{r^2}+\frac{Q^2}{r^3}\right) V_1 \right) f -\frac 1 4   \Box_\g  w_1 \right)|\Phi_1|^2 \\
&& +8\left(\left( - \frac \Up 2  V_2'+\left(-\frac{M}{r^2}+\frac{Q^2}{r^3}\right) V_2 \right) f -\frac 1 4    \Box_\g  w_2 \right)|\Phi_2|^2    \\
  &&+\frac{ 16 Q}{r^3} \left(1-\frac{3M}{r}+\frac{2Q^2}{r^2} +\frac 1 2 \Up\right)f\DDs_2\Phi_1 \c  \Phi_2  
\eeaa
  Observing that $-\frac{M}{r^2}+\frac{Q^2}{r^3}=-\frac 12 \Up'$, we prove the proposition.

   \end{proof}
   
   For $\Phi_1$ and $\Phi_2$ supported in $\ell \geq 2$ spherical harmonics, we use \eqref{elliptic-spj} and \eqref{elliptic-spj-2} to write
     \beaa
\EE^{(Y, w, 0)}[\Phi_1, \Phi_2]&=_{s}& f'|R \Phi_1|^2+\left(\frac{f}{r^3}\left( 1-\frac{3M}{r}+\frac{2Q^2}{r^2}  \right)(\lambda-1)-\frac 1 2 \pr_r\left(\Up V_1\right)f -\frac 1 4   \Box_g  w \right)|\Phi_1|^2   \\
 &&+8 f'|R \Phi_2|^2+8\left(\frac{ f}{r^3}\left( 1-\frac{3M}{r}+\frac{2Q^2}{r^2}  \right)(\lambda-4)-\frac 1 2 \pr_r\left(\Up V_2\right)f -\frac 1 4    \Box_g  w \right)|\Phi_2|^2    \\
  &&+\frac{ 16 Q}{r^3} f\left(1-\frac{3M}{r}+\frac{2Q^2}{r^2} +\frac 1 2 \Up\right)\DDs_2\Phi_1 \c  \Phi_2
\eeaa
for $\lambda=\ell(\ell+1) \geq 6$. Denote 
 \bea
 A_1&:=&\frac{f}{r^3}\left( 1-\frac{3M}{r}+\frac{2Q^2}{r^2}  \right)(\lambda-1) -\frac 1 2 \pr_r\left(\Up V_1\right)f -\frac 1 4   \Box_g  w \label{definition-A-1}\\
 A_2&:=&\frac{f}{r^3}\left( 1-\frac{3M}{r}+\frac{2Q^2}{r^2}  \right)(\lambda-4) -\frac 1 2 \pr_r\left(\Up V_2\right)f -\frac 1 4    \Box_g  w \label{definition-A-2}
 \eea
 the coefficients of $|\Phi_1|^2$ and $|\Phi_2|^2$ respectively.

 In order to have positivity of $\EE^{(Y, w, 0)}[\Phi_1, \Phi_2]$ we need the following {\textit{necessary}} conditions in the exterior region:
 \begin{itemize}
 \item {\bf{Condition 1:}} Positivity of the coefficients of the angular derivatives $|\nabb \Phi_1|^2$ and $|\nabb \Phi_2|^2$, i.e. $ f\left( 1-\frac{3M}{r}+\frac{2Q^2}{r^2}  \right) \geq 0$, \\
 \item {\bf{Condition 2:}} Positivity of the coefficients of the radial derivatives $|R \Phi_1|^2$ and $|R \Phi_2|^2$, i.e. $f' > 0$, \\
 \item {\bf{Condition 3:}} Positivity of the coefficients of $|\Phi_1|^2$ and $|\Phi_2|^2$, i.e. $A_1> 0$ and $A_2 > 0$
 \end{itemize}

 We now consider {\textit{sufficient}}  conditions for the positivity of the quadratic form $\EE^{(Y, w, 0)}[\Phi_1, \Phi_2]$. Using \eqref{dot-product-estimate-l} to bound the mixed term, we have
  \beaa
\EE^{(Y, w, 0)}[\Phi_1, \Phi_2]&\geq_s&  f'|R \Phi_1|^2+8 f'|R \Phi_2|^2+A_1 |\Phi_1|^2  +8A_2 |\Phi_2|^2    \\
  &&-\frac{8 Q}{r^4} f \left|1-\frac{3M}{r}+\frac{2Q^2}{r^2} +\frac 1 2 \Up \right|(2 \lambda-4)^{1/2}|\Phi_1||\Phi_2|
\eeaa
Neglecting the terms in $R$ derivative in virtue of Condition 2, the discriminant of the quadratic terms is 
\beaa
\frac{-D}{32}&=&A_1 A_2- \frac{2 Q^2}{r^8} f^2 \left(1-\frac{3M}{r}+\frac{2Q^2}{r^2} +\frac 1 2 \Up \right)^2(2 \lambda-4)
\eeaa
Writing $\lambda=(\lambda-6)+6$, we have 
\beaa
\frac{-D}{32}&=& \left(\frac{ f}{r^3}\left( 1-\frac{3M}{r}+\frac{2Q^2}{r^2}  \right)(\lambda-6)+B_1\right)\left(\frac{f}{r^3}\left( 1-\frac{3M}{r}+\frac{2Q^2}{r^2}  \right)(\lambda-6)+B_2\right)\\
&&- \frac{4 Q^2}{r^8} f^2 \left(1-\frac{3M}{r}+\frac{2Q^2}{r^2} +\frac 1 2 \Up \right)^2(\lambda-6)- \frac{16 Q^2}{r^8} f^2 \left(1-\frac{3M}{r}+\frac{2Q^2}{r^2} +\frac 1 2 \Up \right)^2
\eeaa
where
\bea
B_1&:=& \frac{5 f}{r^3}\left( 1-\frac{3M}{r}+\frac{2Q^2}{r^2}  \right)-\frac 1 2 \pr_r\left(\Up V_1\right)f -\frac 1 4   \Box_g  w \label{def-B1}\\
B_2&:=& \frac{2f}{r^3}\left( 1-\frac{3M}{r}+\frac{2Q^2}{r^2}  \right)-\frac 1 2 \pr_r\left(\Up V_2\right)f -\frac 1 4   \Box_g  w \label{def-B2}
\eea
In particular, 
\beaa
\frac{-D}{32}&=& \frac{ f^2}{r^6}\left( 1-\frac{3M}{r}+\frac{2Q^2}{r^2}  \right)^2(\lambda-6)^2\\
&&+ B_1 B_2- \frac{16 Q^2}{r^8} f^2 \left(1-\frac{3M}{r}+\frac{2Q^2}{r^2} +\frac 1 2 \Up \right)^2\\
&&+(\lambda-6)\Big[ \frac{ f}{r^3}\left( 1-\frac{3M}{r}+\frac{2Q^2}{r^2}  \right) (B_1+B_2)- \frac{4 Q^2}{r^8} f^2 \left(1-\frac{3M}{r}+\frac{2Q^2}{r^2} +\frac 1 2 \Up \right)^2\Big]
\eeaa

In order to have positivity of the discriminant, we have the following {\textit{sufficient}}\footnote{Conditions 4 and 5 are not necessary. For example, one could use the positivity of the $R$ derivative to absorb part of the mixed term for high spherical harmonics. Nevertheless, we prefer to have a unique approach to all frequencies. } conditions in the exterior region:
\begin{itemize}
\item  {\bf{Condition 4:}} Positivity of the third line in the above expression of the discriminant, i.e. $D_1:= \frac{ f}{r^3}\left( 1-\frac{3M}{r}+\frac{2Q^2}{r^2}  \right) (B_1+B_2)- \frac{4 Q^2}{r^8} f^2 \left(1-\frac{3M}{r}+\frac{2Q^2}{r^2} +\frac 1 2 \Up \right)^2 \geq 0$, 
\item  {\bf{Condition 5:}} Positivity of the second line in the above expression of the discriminant, i.e.  $D_2:=B_1 B_2- \frac{16 Q^2}{r^8} f^2 \left(1-\frac{3M}{r}+\frac{2Q^2}{r^2} +\frac 1 2 \Up \right)^2 \geq 0$.
\end{itemize}

If Conditions 1, 2, 3, 4 and 5 are satisfied, then the bulk integral $ \int_{\MM(\tau_1, \tau_2)} \EE^{(Y, w, 0)}[\Phi_1, \Phi_2]$ is a positive definite form containing derivatives of $\Phi_1$ and $\Phi_2$, with degeneracy at the photon sphere for the angular derivative. Through a standard procedure, we can add control of the $T$ derivative with degeneracy at the photon sphere (see for example Proposition 8 in \cite{Giorgi4}). Recalling definition \eqref{def:Mor-bulk} of $\Mor[\Phi_1, \Phi_2](\tau_1, \tau_2)$, we obtain the following.
 \begin{proposition} Let $\Phi_1$ and $\Phi_2$ be a 1-tensor and a symmetric traceless 2-tensor respectively, satisfying the system of coupled wave equations \eqref{box-phi-1} and \eqref{box-phi-2} and supported in $\ell \geq 2$ spherical harmonics.  If Conditions 1, 2, 3, 4 and 5 hold, then 
 \bea
 \int_{\MM(\tau_1, \tau_2)} \EE^{(Y, w, 0)}[\Phi_1, \Phi_2] \ges  \Mor[\Phi_1, \Phi_2](\tau_1, \tau_2)
 \eea
 \end{proposition}
 
 Our goal is to define functions $f$ and $w$, related by $w=  r^{-2}\Up\pr_r \left( r^2 f(r)\right)  $,  which verify Conditions 1, 2, 3, 4 and 5 in the whole exterior region of the spacetime. This is done in the following subsection.

 Once those conditions are proved to hold, by adding the Morawetz estimates to the energy estimates obtained in Section \ref{sec-en}, i.e. considering the triplet
\beaa
(X, w, M):=(Y, w, 0)+\Lambda (T, 0, 0)
\eeaa
for $\Lambda$ big enough, we obtain 
\bea\label{final-morawets}
E[\Phi_1, \Phi_2](\tau)+\Mor[\Phi_1, \Phi_2] (0, \tau) \les E[\Phi_1, \Phi_2](0) \qquad \text{for any $\tau \geq 0$}
\eea

In the remaining of this Section we will prove that  Conditions 1, 2, 3, 4 and 5 are satisfied in the whole exterior region for subextremal Reissner-Nordstr\"om spacetimes.

 \subsection{Construction of the functions $w$ and $f$}
 
Consider Reissner-Nordstr\"om spacetimes with $|Q|<  M$. We collect here the following facts: 
 \begin{itemize}
 \item The event horizon $r_\mathcal{H}=M+\sqrt{M^2-Q^2}$ ranges between $ M < r_\mathcal{H} \leq 2M$. Observe that since we are restricting our analysis in the exterior region, where $r \geq r_\mathcal{H}$, we always have $\Up=1-\frac{2M}{r}+\frac{Q^2}{r^2} \geq 0$. 
 \item The photon sphere $r_P=\frac{3M+\sqrt{9M^2-8Q^2}}{2}$ ranges between $2 M < r_P \leq 3M$. Observe that $r \geq r_P$ if and only if $1-\frac{3M}{r}+\frac{2Q^2}{r^2}\geq 0$.  
 \end{itemize}

 We also define the following notable points which are used in the construction of the Morawetz functions. 
 \begin{itemize}
 \item $r_1:=\frac{5M+\sqrt{25M^2-24Q^2}}{4}<r_P$, which ranges between $\frac{3M}{2} < r_1 \leq \frac{5M}{2}$. Observe that $r \geq r_1$ if and only if $ 1-\frac{5M}{2r}+\frac{3Q^2}{2r^2}\geq 0$. 
 \item $r_2:=2M +\sqrt{4M^2-3 Q^2}>r_P $, which ranges between $3M < r_2 \leq 4M$. Observe that $r \geq r_2$ if and only if $1-\frac{4M}{r}+\frac{3Q^2}{r^2} \geq 0$. 
 \end{itemize}
Inspired\footnote{Our definition of $w$ \eqref{definitionw-2} differs from \cite{stabilitySchwarzschild} in that there $w$ is defined separately in two intervals, as opposed to three, and in one of them $w=\frac{2}{r} \Up$, as opposed to $w=\frac{2}{r^2} \Up$. We modified it in order to obtain positivity of the bulk in the exterior region for the full subextremal range $|Q|<M$.  The definition of $f$ in terms of $w$ is identical to \cite{Stogin} and \cite{stabilitySchwarzschild}.} by \cite{Stogin} and \cite{stabilitySchwarzschild}, we define $w$ as follows:
 \bea \label{definitionw-2}
 w=\begin{cases}
 &  \frac{ 2}{ r_1^2} \Up(r_1):=w_1>0  \qquad  \quad  \mbox{if}  \quad    r< r_1\\
  & \frac {2}{ r^2} \Up(r)  \qquad \qquad \qquad\mbox{if} \, \,\,\quad    r_1 \le r\le r_2 \\
  & \frac{ 2}{ r_2^2} \Up(r_2):=w_2>0  \qquad  \quad  \mbox{if}  \quad    r> r_2
 \end{cases}
 \eea 
 From the condition $w=  r^{-2}\Up\pr_r \left( r^2 f\right)  $ we  define as in \cite{Stogin} and \cite{stabilitySchwarzschild}:
 \bea
 \label{eq:John2-1}
  r^2 f= \int_{r_P} ^r  \frac{r^2}{\Up} w
   \eea
   
   Since $w>0$, by \eqref{eq:John2-1} we notice that $f$ changes sign at the photon sphere. Condition 1 is then always satisfied. 
   
 In what follows,  we will show that Conditions 2, 3, 4 and 5 are satisfied in the whole exterior region. We separate the analysis in the three regions of the spacetime according to the definition of $w$.    
   
   \subsection{The regions $r \leq r_1$ and $r \geq r_2$}
   
   In the regions  $r \leq r_1$ and $r \geq r_2$, the function $w$ is a positive constant. In particular we have
   \beaa
   w=w_1, \qquad f=\frac{w_1}{r^2} \int_{r_P}^r \frac{r^2}{\Up} \qquad \text{for $r\leq r_1$}\\
    w=w_2, \qquad f=\frac{w_2}{r^2} \int_{r_P}^r \frac{r^2}{\Up} \qquad \text{for $r\geq r_2$}
   \eeaa
    
   We start by proving that Condition 2 is satisfied in these two regions.

     \begin{lemma}\label{lemma-positivity-f'} Condition 2 is verified for $r \leq r_1$ and for $r\geq r_2$, i.e. for all subextremal Reissner-Nordstr\"om spacetimes with $|Q|<M$ we have
  \beaa
  f'(r)>0 \qquad \text{for $ r \leq  r_1$ or $r \geq r_2$}
  \eeaa
  \end{lemma}
  \begin{proof} 
   Since $f$ vanishes at $r=r_P$ we have,
\beaa
f' &=& r^{-2} (r^2 f)' - 2 r^{-3}(r^2 f)= r^{-2} (r^2 f)' -2 r^{-3}\int_{r_P}^{r} (r^2 f)' 
\eeaa
Recall that $w=r^{-2}\Up (r^2 f)'=w_1>0$ for $r\leq r_1$, and $w=r^{-2}\Up (r^2 f)'=w_2>0$ for $r \geq r_2$. We deduce
\beaa
(r^2 f)'=\frac{r^2 }{\Up}w_1 \qquad \text{for $r \leq r_1$,}  \qquad (r^2 f)'=\frac{r^2 }{\Up}w_2 \qquad \text{for $r \geq r_2$}
\eeaa
Consider the region $r \leq r_1$. We write
\beaa
f' &=&  \frac{1 }{\Up}w_1 +2w_1 r^{-3}\int_{r}^{r_P} \frac{r^2 }{\Up}=w_1\left(\frac{1 }{\Up}+2 r^{-3}\int_{r}^{r_P} \frac{r^2 }{\Up}\right)
\eeaa
Observe that the function $\frac{r^2}{\Up}$ is increasing if $r>r_P$ and decreasing otherwise. Indeed, 
\bea\label{derivative-r-2-Up}
\pr_r\left(\frac{r^2}{\Up} \right)&=& \frac{1}{\Up^2} \left(2r \Up-r^2 \Up' \right)=\frac{2r}{\Up^2} \left(1-\frac{3M}{r}+\frac{2Q^2}{r^2} \right)
\eea
The integrand $\frac{r^2}{\Up}$ is therefore decreasing for $r \leq r_1 < r_P$. 
We bound the integral of a decreasing function over an interval by the value of the function at the right end of the interval times the length of the interval. Similarly $\frac{1}{\Up}$ is everywhere decreasing. We obtain
\beaa
f' &>&w_1\left(\frac{1 }{\Up(r_P)}+2 r_P^{-3}(r_P-r) \frac{r_P^2 }{\Up(r_P)}\right)=\frac{w_1}{r_P\Up(r_P)}\left(3 r_P-2r \right)
\eeaa
which is positive for $r \leq  r_1 < r_P < \frac 3 2 r_P$.

Consider now the region $r \geq r_2$. We have 
\beaa
f' &=&w_2 \left(  \frac{1}{\Up}  -2  r^{-3}\int_{r_P}^{r} \frac{r^2}{\Up} \right)
\eeaa
The integrand $\frac{r^2}{\Up}$ is increasing for $r \geq r_2$. We bound the integral of an increasing function from above by the value of the function at the right end of the interval times the length of the interval, i.e. $\int_{r_P}^{r} \frac{r^2}{\Up} \leq \frac{r^2}{\Up}(r-r_P)$. 
This gives
\beaa
f' &\geq &w_2 \left(  \frac{1}{\Up}  -2  \frac{1}{r\Up}(r-r_P) \right)=\frac{w_2}{r\Up} \left(  -r  +2  r_P \right)
\eeaa
which is positive for $r < 2r_P$, which contains $r=r_2$. In particular, $(r^3 f')|_{r=r_2}>0$. 
Using \eqref{eq:John2-1} written as $\pr_r(r^2 f)=\frac{r^2}{\Up} w$, we have 
\beaa
 r^2 \pr_r \left(r\frac{w}{\Up}\right)&=&r^2 \left[ r^{-1} \pr_r \left(\frac{r^2}{ \Up} w \right)- r^{-2} \frac{r^2}{ \Up} w \right]=  r  \pr_r \left(\frac{r^2}{ \Up} w \right)- \frac{r^2}{ \Up} w\\
 &=& r \pr_r\pr_r(r^2 f)-\pr_r(r^2 f) = 3 r^2 f'+ r^3 f''=\pr_r( r^3 f')
 \eeaa
Therefore
   \beaa
    \pr_r(r^3 f' )=r^2 \pr_r \left(r\frac{w}{\Up}\right)   = w_2 r^2\pr_r\left (\frac{r}{ \Up}\right) 
    \eeaa
Since
   \beaa
   \pr_r\left(\frac{r}{\Up} \right)&=& \frac{1}{\Up^2} \left( \Up-r \Up' \right)= \frac{1}{\Up^2} \left( 1-\frac{4M}{r}+\frac{3Q^2}{r^2}\right)
   \eeaa
    we have
   \beaa
      \pr_r(r^3 f' ) \geq 0 \qquad \text{for $r \geq r_2$}
   \eeaa
   This implies $r^3 f' \geq (r^3 f')|_{r=r_2}>0$. This concludes the proof of the lemma. 

  \end{proof}

We now prove that Condition 3 is verified, i.e. the coefficients of the zeroth order terms are positive. 

  \begin{lemma}\label{lemmafirst-cond3-int} Condition 3 is verified for $r \leq r_1$ and for $r \geq r_2$, i.e. for all subextremal Reissner-Nordstr\"om spacetimes with $|Q| < M$ we have 
\beaa
 A_1 \geq 0, \qquad A_2 \geq 0 \qquad \text{ for $r \leq r_1$ or $r \geq r_2$}
\eeaa
\end{lemma}
   \begin{proof} Because of Condition 1, the first term in the definition of $A_1$ and $A_2$ in \eqref{definition-A-1} and \eqref{definition-A-2} is always positive. For $r \leq r_1$ or $r \geq r_2$, $w$ is constant, therefore $\square_g w=0$.  This gives
   \beaa
   A_1 \geq -\frac 1 2 \pr_r\left(\Up V_1\right) f \qquad      A_2 \geq -\frac 1 2 \pr_r\left(\Up V_2\right) f
   \eeaa

   Since $f \leq 0$ for $ r \leq r_P$ and $f \geq 0$ for $r \geq r_P$, we are only left to prove that for $r \leq r_1$ then  $-\frac 1 2 \pr_r\left(\Up V_1\right), -\frac 1 2 \pr_r\left(\Up V_2\right) \leq 0$  and  for $r \geq r_2$ then $-\frac 1 2 \pr_r\left(\Up V_1\right), -\frac 1 2 \pr_r\left(\Up V_2\right) \geq 0$.

   Recall that 
\bea\label{potentials-derivatives}
\begin{split}
V_1&=\frac{1}{r^2}\left(1-\frac{2M}{r}+\frac{6Q^2}{r^2} \right)=\frac{1}{r^2}\Up+\frac{5Q^2}{r^4} ,\quad  V_2=\frac{4}{r^2}\left( 1-\frac{2M}{r}+\frac{3Q^2}{2r^2}\right)=\frac{4}{r^2} \Up+\frac{2Q^2}{r^4} \\
V'_1&=-\frac{2}{r^3}\left(1-\frac{3M}{r}+\frac{12Q^2}{r^2} \right), \qquad V'_2=-\frac{8}{r^3}\left( 1-\frac{3M}{r}+\frac{3Q^2}{r^2}\right) 
\end{split}
\eea
Since $\pr_r \Up=\frac{2M}{r^2}-\frac{2Q^2}{r^3}$, we compute
\beaa
 \pr_r\left(\Up V_1\right)&=&\left(\frac{2M}{r^2}-\frac{2Q^2}{r^3}\right) \left(\frac{1}{r^2}\Up+\frac{5Q^2}{r^4}  \right)-\Up \frac{2}{r^3}\left(1-\frac{3M}{r}+\frac{12Q^2}{r^2} \right)\\
 &=&-\frac{2}{r^3}\left(1-\frac{4M}{r}+\frac{13Q^2}{r^2}   \right)\Up+\frac{10Q^2}{r^4}\left(\frac{M}{r^2}-\frac{Q^2}{r^3} \right)\\
  \pr_r\left(\Up V_2\right)&=&-\frac{8}{r^3}\left(1-\frac{4M}{r}+\frac{4Q^2}{r^2}\right) \Up+\frac{4Q^2}{r^4}\left(\frac{M}{r^2}-\frac{Q^2}{r^3}\right)
   \eeaa
Consider the region $r\leq r_1$.   Writing in the above expressions respectively
\beaa
1-\frac{4M}{r}+\frac{13Q^2}{r^2}  &=& \left(1-\frac{3M}{r}+\frac{2Q^2}{r^2} \right)-\left(\frac{M}{r}-\frac{Q^2}{r^2} \right)+\frac{10Q^2}{r^2} \\
1-\frac{4M}{r}+\frac{4Q^2}{r^2}  &=& \left(1-\frac{3M}{r}+\frac{2Q^2}{r^2} \right)-\left(\frac{M}{r}-\frac{Q^2}{r^2} \right)+\frac{Q^2}{r^2} 
\eeaa
 we obtain 
  \bea\label{potenti-r1}
  \begin{split}
-\frac 1 2 \pr_r\left(\Up V_1\right) &=\frac{1}{r^3}\left(1-\frac{4M}{r}+\frac{13Q^2}{r^2}   \right)\Up-\frac{5Q^2}{r^4}\left(\frac{M}{r^2}-\frac{Q^2}{r^3} \right)\\
&= \frac{1}{r^3}\left(1-\frac{3M}{r}+\frac{2Q^2}{r^2} \right)\Up -\frac{1}{r^3}\left(\frac{M}{r}-\frac{Q^2}{r^2} \right)\Up+\frac{5Q^2}{r^5}  \left( 2\Up-\frac{M}{r}+\frac{Q^2}{r^2} \right)
\end{split}
\eea
and
\bea\label{potenti-r2}
\begin{split}
-\frac 1 2   \pr_r\left(\Up V_2\right)&=\frac{4}{r^3}\left(1-\frac{4M}{r}+\frac{4Q^2}{r^2}\right) \Up-\frac{2Q^2}{r^4}\left(\frac{M}{r^2}-\frac{Q^2}{r^3}\right)\\
&=\frac{4}{r^3} \left(1-\frac{3M}{r}+\frac{2Q^2}{r^2} \right)\Up -\frac{4}{r^3} \left(\frac{M}{r}-\frac{Q^2}{r^2} \right)\Up+\frac{2Q^2}{r^5}  \left( 2\Up-\frac{M}{r}+\frac{Q^2}{r^2} \right)
\end{split}
   \eea
   The first term in both expressions is negative since $r\leq r_1 <r_P$. The second term can be written as $   -\frac{1}{r^5} \left(Mr-Q^2 \right)\Up$, 
   and since $r_\mathcal{H} > M$, we have $Mr-Q^2>M^2-Q^2>0$. This proves the negativity of the second term in both expressions. The last term is given by $2\Up-\frac{M}{r}+\frac{Q^2}{r^2}= 2\left(1-\frac{5M}{2r}+\frac{3Q^2}{2r^2}\right)$, 
which is negative since $r \leq r_1$.

   Consider the region $r \geq r_2$. Writing  in the expressions above respectively 
   \beaa
   1-\frac{4M}{r}+\frac{13Q^2}{r^2}  &=&\left(1-\frac{4M}{r}+\frac{3Q^2}{r^2}\right) +\frac{10Q^2}{r^2}, \qquad 1-\frac{4M}{r}+\frac{4Q^2}{r^2}  = \left(1-\frac{4M}{r}+\frac{3Q^2}{r^2}\right)+\frac{Q^2}{r^2}
   \eeaa
   we obtain 
     \bea\label{potential-r21}
     \begin{split}
-\frac 1 2 \pr_r\left(\Up V_1\right) &=\frac{1}{r^3}\left(1-\frac{4M}{r}+\frac{13Q^2}{r^2}   \right)\Up-\frac{5Q^2}{r^4}\left(\frac{M}{r^2}-\frac{Q^2}{r^3} \right)\\
&=\frac{1}{r^3}\left(1-\frac{4M}{r}+\frac{3Q^2}{r^2}     \right)\Up+\frac{5Q^2}{r^5} \left(2\Up-\frac{M}{r}+\frac{Q^2}{r^2} \right)
\end{split}
\eea
and
\bea\label{potential-r22}
\begin{split}
-\frac 1 2   \pr_r\left(\Up V_2\right)&=\frac{4}{r^3}\left(1-\frac{4M}{r}+\frac{4Q^2}{r^2}\right) \Up-\frac{2Q^2}{r^4}\left(\frac{M}{r^2}-\frac{Q^2}{r^3}\right)\\
&=\frac{4}{r^3}\left(1-\frac{4M}{r}+\frac{3Q^2}{r^2}\right) \Up+\frac{2Q^2}{r^5} \left(2\Up-\frac{M}{r}+\frac{Q^2}{r^2} \right)
\end{split}
   \eea
The first term in the both expression is positive for $r\geq r_2$ and the second term is positive for $r\geq r_2 > r_1$. This concludes the proof of the lemma.

   \end{proof}
   
   We now prove that Condition 4 is verified, i.e. that $D_1$ is positive.

     \begin{lemma}\label{lemmafirst-cond4-int} Condition 4 is verified for $r \leq r_1$ and for $r\geq r_2$, i.e. for all subextremal Reissner-Nordstr\"om spacetimes with $|Q| <  M$ we have 
\beaa
D_1 \geq 0 \qquad \text{ for $r \leq r_1$ or $r\geq r_2$}
\eeaa
\end{lemma}
\begin{proof} Recall the definitions \eqref{def-B1} and \eqref{def-B2} of $B_1$ and $B_2$.  Then 
\beaa
B_1+B_2= f \left(\frac{7}{r^3}\left( 1-\frac{3M}{r}+\frac{2Q^2}{r^2}  \right) -\frac 1 2 \pr_r\left(\Up V_1\right) -\frac 1 2 \pr_r\left(\Up V_2\right) \right)
\eeaa
 therefore the expression for $D_1$ reads
\beaa
D_1&=& \frac{ f}{r^3}\left( 1-\frac{3M}{r}+\frac{2Q^2}{r^2}  \right) (B_1+B_2)- \frac{4 Q^2}{r^8} f^2 \left(1-\frac{3M}{r}+\frac{2Q^2}{r^2} +\frac 1 2 \Up \right)^2\\
 &=& \frac{ f^2}{r^3}\left( 1-\frac{3M}{r}+\frac{2Q^2}{r^2}  \right) \left( \frac{7}{r^3}\left( 1-\frac{3M}{r}+\frac{2Q^2}{r^2}  \right) -\frac 1 2 \pr_r\left(\Up V_1\right) -\frac 1 2 \pr_r\left(\Up V_2\right)\right)\\
 &&- \frac{4 Q^2}{r^8} f^2 \left(1-\frac{3M}{r}+\frac{2Q^2}{r^2} +\frac 1 2 \Up \right)^2
\eeaa
and writing
\beaa
&&- \frac{4 Q^2}{r^8} f^2 \left(1-\frac{3M}{r}+\frac{2Q^2}{r^2} +\frac 1 2 \Up \right)^2\\
&=& - \frac{4 Q^2}{r^8} f^2 \left(1-\frac{3M}{r}+\frac{2Q^2}{r^2}  \right)^2- \frac{4 Q^2}{r^8} f^2 \left(1-\frac{3M}{r}+\frac{2Q^2}{r^2}\right) \Up - \frac{ Q^2}{r^8} f^2  \Up^2
\eeaa
the above becomes
\bea\label{cond-4-reads}
\begin{split}
D_1=&\frac{f^2}{r^6}\left( 7-\frac{4 Q^2}{r^2}\right) \left( 1-\frac{3M}{r}+\frac{2Q^2}{r^2}  \right)^2 +\frac{ f^2}{r^3}\left( 1-\frac{3M}{r}+\frac{2Q^2}{r^2}  \right) \left(  -\frac 1 2 \pr_r\left(\Up V_1\right) -\frac 1 2 \pr_r\left(\Up V_2\right)\right)\\
&- \frac{4 Q^2}{r^8} f^2 \left(1-\frac{3M}{r}+\frac{2Q^2}{r^2}\right) \Up - \frac{ Q^2}{r^8} f^2  \Up^2
\end{split}
\eea
We first consider the region $r \leq r_1$. In this region we use \eqref{potenti-r1} and \eqref{potenti-r2} to write
\beaa
-\frac 1 2 \pr_r\left(\Up V_1\right) -\frac 1 2 \pr_r\left(\Up V_2\right) &=& \frac{5\Up}{r^3}\left( 1-\frac{3M}{r}+\frac{2Q^2}{r^2}  \right) -\frac{5}{r^3}\left(\frac{M}{r}-\frac{Q^2}{r^2} \right)\Up+\frac{7Q^2}{r^5}  \left( 2\Up-\frac{M}{r}+\frac{Q^2}{r^2} \right)
\eeaa
Factorizing out the positive factor $\frac{f^2}{r^6}$, \eqref{cond-4-reads} becomes
\beaa
&&\left( 7+5\Up -\frac{4 Q^2}{r^2}\right) \left( 1-\frac{3M}{r}+\frac{2Q^2}{r^2}  \right)^2  +  \frac{ 5\Up }{r^2}\left( -1+\frac{3M}{r}-\frac{2Q^2}{r^2}  \right) \left(Mr-Q^2\right) \\
&& \frac{7Q^2}{r^2} \left( 1-\frac{3M}{r}+\frac{2Q^2}{r^2}  \right)  \left( 2\Up-\frac{M}{r}+\frac{Q^2}{r^2} \right)   +\frac{4 Q^2}{r^2}\Up\left(-1+\frac{3M}{r}-\frac{2Q^2}{r^2}\right)   -\frac{ Q^2}{r^2} \Up^2
\eeaa
The first term is positive because it can be bounded from below by $7- \frac{4Q^2}{r^2}=\frac{1}{r^2} (7r^2-4Q^2)$, which is always positive for $r\geq r_\mathcal{H} >M$. All the other terms, except the last one, are clearly positive for $ r \leq r_1$. 
Observe that for $ r\leq r_1$, we have the following bounds:
\bea
1-\frac{2M}{r}+\frac{Q^2}{r^2} &=& \left( 1-\frac{5M}{2r}+\frac{3Q^2}{2r^2}\right)+\frac{M}{2r}-\frac{Q^2}{2r^2}\leq \frac {1}{ 2r^2}\left( Mr-Q^2\right)\label{bound-h-close} \\
-1+\frac{3M}{r}-\frac{2Q^2}{r^2}&=&\left( -1+\frac{5M}{2r}-\frac{3Q^2}{2r^2}\right)+\frac{M}{2r}-\frac{Q^2}{2r^2}\geq \frac {1}{ 2r^2}\left( Mr-Q^2\right)\label{bound-ph-close}
\eea
The last two terms together are positive:  indeed, using \eqref{bound-h-close} and \eqref{bound-ph-close}, we obtain
\beaa
\frac{4 Q^2}{r^8}\Up\left(-1+\frac{3M}{r}-\frac{2Q^2}{r^2}\right) -\frac{ Q^2}{r^8} \Up^2  &\geq&\frac{4 Q^2}{r^8}\Up\left(\frac {1}{ 2r^2}\left( Mr-Q^2\right) \right)  -\frac{ Q^2}{r^8} \Up^2   \\
&=&   \frac{ Q^2}{r^{8}}\Up\left(\frac{2}{r^2}\left( Mr-Q^2\right)-\Up \right) \geq  \frac{ Q^2}{r^{8}}\Up\left(\frac{3}{2r^2}\left( Mr-Q^2\right)\right) 
\eeaa
which proves Condition 4 for $r \leq r_1$. 

We now consider the region $r \geq r_2$. In this region we use \eqref{potential-r21} and \eqref{potential-r21} to write
\beaa
-\frac 1 2 \pr_r\left(\Up V_1\right) -\frac 1 2 \pr_r\left(\Up V_2\right) &=&5\left(1-\frac{4M}{r}+\frac{3Q^2}{r^2}     \right)\Up+\frac{7Q^2}{r^2} \left(2\Up-\frac{M}{r}+\frac{Q^2}{r^2} \right)
\eeaa
Factorizing out the positive factor $\frac{f^2}{r^6}$, \eqref{cond-4-reads} becomes
\beaa
&&\left(7- \frac{4 Q^2}{r^2} \right)\left( 1-\frac{3M}{r}+\frac{2Q^2}{r^2}  \right)^2+5\Up\left( 1-\frac{3M}{r}+\frac{2Q^2}{r^2}  \right) \left(1-\frac{4M}{r}+\frac{3Q^2}{r^2}     \right)\\
&&+\frac{7Q^2}{r^2}\left( 1-\frac{3M}{r}+\frac{2Q^2}{r^2}  \right)  \left(2\Up-\frac{M}{r}+\frac{Q^2}{r^2} \right)- \frac{ Q^2}{r^2}   \Up^2- \frac{4 Q^2}{r^2}  \left(1-\frac{3M}{r}+\frac{2Q^2}{r^2}\right)  \Up 
\eeaa
The first line is clearly positive. The second line can be arranged to be
\beaa
\frac{Q^2}{r^2}\left( 1-\frac{3M}{r}+\frac{2Q^2}{r^2}  \right)  \left(10\Up-\frac{7M}{r}+\frac{7Q^2}{r^2} \right)- \frac{ Q^2}{r^2}   \Up^2
\eeaa
Writing $ 1-\frac{3M}{r}+\frac{2Q^2}{r^2}=\Up -\frac{M}{r}+\frac{Q^2}{r^2}$ and factorizing out $\frac{Q^2}{r^2}$, we have 
\beaa
&&\left( \Up -\frac{M}{r}+\frac{Q^2}{r^2}  \right)  \left(10\Up-\frac{7M}{r}+\frac{7Q^2}{r^2} \right)-    \Up^2= 9\Up^2+17\Up\left(  -\frac{M}{r}+\frac{Q^2}{r^2}  \right)  +7\left( -\frac{M}{r}+\frac{Q^2}{r^2}  \right)^2\\
&\geq & \Up \left( 9\Up+17\Up\left(  -\frac{M}{r}+\frac{Q^2}{r^2}  \right)  \right)
\eeaa
Since for $ r\geq r_2$, we have 
 \bea\label{estimate-Up-r2}
 \Up= 1-\frac{2M}{r}+\frac{Q^2}{r^2} =1-\frac{4M}{r}+\frac{3Q^2}{r^2} +\frac{2M}{r}-\frac{2Q^2}{r^2}  \geq  \frac{2}{r^2} \left(Mr-Q^2 \right)
 \eea
 the above is positive: $ 9\Up+17\Up\left(  -\frac{M}{r}+\frac{Q^2}{r^2}  \right)  \geq  9\frac{2}{r^2} \left(Mr-Q^2 \right)+17\Up\left(  -\frac{M}{r}+\frac{Q^2}{r^2}  \right) =\frac{1}{r^2} \left(Mr-Q^2 \right)>0$. 
This shows that condition 4 is verified for $r \geq r_2$. This concludes the proof of the lemma.
\end{proof}

We now prove that Condition 5 is verified, i.e. that $D_2$ is positive. 

 \begin{lemma}\label{lemmafirst-cond5-int} Condition 5 is verified for $r \leq r_1$ and for $r \geq r_2$, i.e. for all subextremal Reissner-Nordstr\"om spacetimes with $|Q| < M$ we have 
\beaa
D_2 \geq 0 \qquad \text{ for $r \leq r_1$ or $r\geq r_2$}
\eeaa
\end{lemma}
\begin{proof} Recall the definitions \eqref{def-B1} and \eqref{def-B2} of $B_1$ and $B_2$.  Then Condition 5 reads
 \beaa
 D_2&=&B_1 B_2- \frac{16 Q^2}{r^8} f^2 \left(1-\frac{3M}{r}+\frac{2Q^2}{r^2} +\frac 1 2 \Up \right)^2\\
  &=&f^2 \left( \frac{5}{r^3}\left( 1-\frac{3M}{r}+\frac{2Q^2}{r^2}  \right) -\frac 1 2 \pr_r\left(\Up V_1\right)\right) \left( \frac{2}{r^3}\left( 1-\frac{3M}{r}+\frac{2Q^2}{r^2}  \right) -\frac 1 2 \pr_r\left(\Up V_2\right)\right)\\
  &&- \frac{16 Q^2}{r^8} f^2 \left(1-\frac{3M}{r}+\frac{2Q^2}{r^2} +\frac 1 2 \Up \right)^2
 \eeaa
 The above becomes
 \bea\label{identituy}
 \begin{split}
 D_2&=f^2 \left(\frac{10}{r^6}- \frac{16 Q^2}{r^8}\right)\left( 1-\frac{3M}{r}+\frac{2Q^2}{r^2}  \right)^2\\
 & +f^2 \frac{1}{r^3}\left(2\left( -\frac 1 2 \pr_r\left(\Up V_1\right)\right)+5\left( -\frac 1 2 \pr_r\left(\Up V_2\right)\right)\right) \left( 1-\frac{3M}{r}+\frac{2Q^2}{r^2}  \right)\\
 &+f^2 \left( -\frac 1 2 \pr_r\left(\Up V_1\right)\right) \left(-\frac 1 2 \pr_r\left(\Up V_2\right)\right)- \frac{16 Q^2}{r^8} f^2 \left(1-\frac{3M}{r}+\frac{2Q^2}{r^2}\right) \Up - \frac{ 4Q^2}{r^8} f^2  \Up^2
 \end{split}
 \eea
 We first consider the region $r \leq r_1$. In this region we use \eqref{potenti-r1} and \eqref{potenti-r2} to write
 \beaa
-\pr_r\left(\Up V_1\right) -\frac 5 2 \pr_r\left(\Up V_2\right)&=&\frac{22}{r^3}\left(1-\frac{3M}{r}+\frac{2Q^2}{r^2} \right)\Up -\frac{22}{r^3}\left(\frac{M}{r}-\frac{Q^2}{r^2} \right)\Up +\frac{20Q^2}{r^5}  \left( 2\Up-\frac{M}{r}+\frac{Q^2}{r^2} \right)
 \eeaa
 and 
 \beaa
 && \left( -\frac 1 2 \pr_r\left(\Up V_1\right)\right) \left(-\frac 1 2 \pr_r\left(\Up V_2\right)\right)= \frac{4\Up^2}{r^6}\left(1-\frac{3M}{r}+\frac{2Q^2}{r^2} \right)^2 -\frac{8\Up^2}{r^6}\left(\frac{M}{r}-\frac{Q^2}{r^2} \right) \left(1-\frac{3M}{r}+\frac{2Q^2}{r^2} \right) \\
 &&+\frac{22Q^2\Up}{r^8}  \left( 2\Up-\frac{M}{r}+\frac{Q^2}{r^2} \right) \left(1-\frac{3M}{r}+\frac{2Q^2}{r^2} \right)  \\
 &&+  \frac{4\Up^2}{r^6}\left(\frac{M}{r}-\frac{Q^2}{r^2} \right)^2- \frac{22Q^2\Up}{r^8}  \left( 2\Up-\frac{M}{r}+\frac{Q^2}{r^2} \right) \left(\frac{M}{r}-\frac{Q^2}{r^2} \right)+  \frac{10Q^4}{r^{10}}  \left( 2\Up-\frac{M}{r}+\frac{Q^2}{r^2} \right)^2 
 \eeaa
 Factorizing out $\frac{f^2}{r^6}$, \eqref{identituy} becomes 
 \small
  \beaa
&&\left( 10+22\Up+4\Up^2-\frac{16 Q^2}{r^2}\right)\left(- 1+\frac{3M}{r}-\frac{2Q^2}{r^2}  \right)^2 +\Up \frac{ 22+8\Up}{r^2}\left( -1+\frac{3M}{r}-\frac{2Q^2}{r^2}  \right)\left(Mr-Q^2 \right)\\
&&+\frac{Q^2}{r^2} (20+22\Up) \left( -2\Up+\frac{M}{r}-\frac{Q^2}{r^2} \right)\left(- 1+\frac{3M}{r}-\frac{2Q^2}{r^2}  \right)+ \frac{4 \Up^2}{r^{4}}\left(Mr-Q^2\right)^2-\frac{4 Q^2}{r^2} \Up^2\\
&& +\frac{22Q^2}{r^{4}} \Up \left( -2\Up+\frac{M}{r}-\frac{Q^2}{r^2} \right) \left(Mr-Q^2 \right)+\frac{10Q^4}{r^{4}}  \left( -2\Up+\frac{M}{r}-\frac{Q^2}{r^2} \right)^2 +\frac{16 Q^2}{r^2}\Up\left(-1+\frac{3M}{r}-\frac{2Q^2}{r^2}\right) 
\eeaa
\normalsize
  Using \eqref{bound-ph-close} to bound $ - 1+\frac{3M}{r}-\frac{2Q^2}{r^2}$ in the above, we obtain
   \beaa
   &\geq&\left( \frac{ (5+\Up)(2+4\Up)}{4r^{4}} -\frac{4 Q^2}{r^{6}}\right)\left( Mr-Q^2\right) ^2  +\Up \frac{ 11+4\Up}{r^{4}}\left(Mr-Q^2 \right)^2\\
&&+\frac{Q^2}{r^{4}} (10+11\Up) \left( -2\Up+\frac{M}{r}-\frac{Q^2}{r^2} \right)\left( Mr-Q^2\right) + \frac{4 \Up^2}{r^{4}}\left(Mr-Q^2\right)^2-\frac{4 Q^2}{r^2} \Up^2\\
&& +\frac{22Q^2}{r^{4}} \Up \left( -2\Up+\frac{M}{r}-\frac{Q^2}{r^2} \right) \left(Mr-Q^2 \right)+\frac{10Q^4}{r^{4}}  \left( -2\Up+\frac{M}{r}-\frac{Q^2}{r^2} \right)^2 +\frac{8 Q^2}{r^{4}}\Up\left( Mr-Q^2\right) 
\eeaa
We can arrange the above as
   \beaa
  &=&\left( \frac{ 5+33\Up +18\Up^2}{2r^{4}} -\frac{4 Q^2}{r^{6}}\right)\left( Mr-Q^2\right) ^2 +\frac{4 Q^2}{r^{2}}\Up \left( \frac{2}{r^2}\left( Mr-Q^2\right)  - \Up\right)\\
&&+\frac{Q^2}{r^{4}} (10+33\Up) \left( -2\Up+\frac{M}{r}-\frac{Q^2}{r^2} \right)\left( Mr-Q^2\right) +\frac{10Q^4}{r^{4}}  \left( -2\Up+\frac{M}{r}-\frac{Q^2}{r^2} \right)^2
\eeaa
Using \eqref{bound-h-close} we can bound $\frac{2}{r^2}\left( Mr-Q^2\right)  - \Up \geq \frac{3}{2r^2}\left( Mr-Q^2\right)$, therefore
 \bea\label{proof-positivity-interm}
 \begin{split}
&\geq\left( \frac{ 5+33\Up +18\Up^2}{2r^{4}} -\frac{4 Q^2}{r^{6}}\right)\left( Mr-Q^2\right) ^2 +\frac{6 Q^2}{r^{4}}\Up \left( Mr-Q^2\right)\\
&+\frac{Q^2}{r^{4}} (10+33\Up) \left( -2\Up+\frac{M}{r}-\frac{Q^2}{r^2} \right)\left( Mr-Q^2\right) +\frac{10Q^4}{r^{4}}  \left( -2\Up+\frac{M}{r}-\frac{Q^2}{r^2} \right)^2
\end{split}
\eea
Observe that the second line is always non negative. 
We write the first line as 
\beaa
&&\left( Mr-Q^2\right) \Big[\left( \frac{ 5+33\Up +18\Up^2}{2r^{4}} -\frac{5 Q^2}{2r^{6}}\right)\left( Mr-Q^2\right) +\left(-\frac{6 Q^2}{r^{6}}\right)\left( \frac{Mr}{4}-\frac{Q^2}{4}\right)\\
&&  +\frac{6 Q^2}{r^{4}}\left(1-\frac{2M}{r}+\frac{Q^2}{r^2}\right) \Big]\\
&=&\left( Mr-Q^2\right) \Big[ \left( \frac{ 5+33\Up +18\Up^2}{2r^{4}} -\frac{5 Q^2}{2r^{6}}\right)\left( Mr-Q^2\right)  +\frac{6 Q^2}{r^{4}}\left(1-\frac{9M}{4r}+\frac{5Q^2}{4r^2}\right) \Big]
\eeaa
Observe that the first of the above two terms is always positive since $r^2 -Q^2 \geq 0$ and the second term is positive if $ \frac{9M + \sqrt{81M^2-80Q^2}}{8} \leq r \leq r_1$. In particular using \eqref{proof-positivity-interm} this proves that Condition 5 is verified in  $ \frac{9M + \sqrt{81M^2-80Q^2}}{8} \leq r \leq r_1$. 

If $r \leq \frac{9M + \sqrt{81M^2-80Q^2}}{8}$, we can bound the factor $ -2\Up+\frac{M}{r}-\frac{Q^2}{r^2}$ away from zero, as 
\beaa
 -2\Up+\frac{M}{r}-\frac{Q^2}{r^2}&=& - 2+\frac{5M}{r}-\frac{3Q^2}{r^2}=- 2\left(1-\frac{9M}{4r}+\frac{5Q^2}{4r^2}\right)+\frac{M}{2r}-\frac{Q^2}{2r^2} \geq \frac{1}{2r^2} \left( Mr-Q^2 \right)
\eeaa
From \eqref{proof-positivity-interm}, we therefore obtain
\small
 \beaa
&\geq&\left( \frac{ 5+33\Up +18\Up^2}{2r^{4}} -\frac{4 Q^2}{r^{6}}+\frac{Q^2}{r^{4}} (10+33\Up)\frac{1}{2r^2} \right)\left( Mr-Q^2\right) ^2 +\frac{6 Q^2}{r^{4}}\Up \left( Mr-Q^2\right) +\frac{10Q^4}{4r^{8}}  \left( Mr-Q^2 \right)^2\\
&\geq&\left( \frac{ 5+33\Up +18\Up^2}{2r^{4}} +\frac{ Q^2}{r^{6}}+\frac{33Q^2}{2r^{6}} \Up\right)\left( Mr-Q^2\right) ^2 +\frac{6 Q^2}{r^{4}}\Up \left( Mr-Q^2\right) +\frac{10Q^4}{4r^{8}}  \left( Mr-Q^2 \right)^2
\eeaa
\normalsize
which proves Condition 5 for all $r\leq r_1$.

 We now consider the region $r\geq r_2$. In this region we use \eqref{potential-r21} and \eqref{potential-r22} to write 
 \beaa
&& 2\left( -\frac 1 2 \pr_r\left(\Up V_1\right)\right)+5\left( -\frac 1 2 \pr_r\left(\Up V_2\right)\right)= \frac{22}{r^3}\left(1-\frac{4M}{r}+\frac{3Q^2}{r^2}     \right)\Up+\frac{20Q^2}{r^5} \left(2\Up-\frac{M}{r}+\frac{Q^2}{r^2} \right)
 \eeaa
 and 
 \beaa
 \left( -\frac 1 2 \pr_r\left(\Up V_1\right)\right) \left(-\frac 1 2 \pr_r\left(\Up V_2\right)\right)&=& \frac{4\Up^2}{r^6}\left(1-\frac{4M}{r}+\frac{3Q^2}{r^2}     \right)^2\\
 &&+\frac{22Q^2}{r^8}\left(1-\frac{4M}{r}+\frac{3Q^2}{r^2}     \right) \left(2\Up-\frac{M}{r}+\frac{Q^2}{r^2} \right)\\
 &&+\frac{10Q^4}{r^{10}} \left(2\Up-\frac{M}{r}+\frac{Q^2}{r^2} \right)^2
 \eeaa
 Factorizing out $\frac{f^2}{r^6}$ and neglecting the terms $1-\frac{4M}{r}+\frac{3Q^2}{r^2} \geq 0$,  \eqref{identituy} becomes 
  \beaa
  &&  \left(10- \frac{16 Q^2}{r^2}\right)\left( 1-\frac{3M}{r}+\frac{2Q^2}{r^2}  \right)^2\\
  &&+  \frac{20Q^2}{r^2} \left(2\Up-\frac{M}{r}+\frac{Q^2}{r^2} \right) \left( 1-\frac{3M}{r}+\frac{2Q^2}{r^2}  \right)- \frac{16 Q^2}{r^2}  \left(1-\frac{3M}{r}+\frac{2Q^2}{r^2}\right) \Up - \frac{4 Q^2}{r^2} \Up ^2
 \eeaa
 The term $10- \frac{16 Q^2}{r^2}$ is positive for $r\geq r_P$ in the subextremal range. The second line can be arranged to be, factorizing out  $\frac{ 4Q^2}{r^2}  $, 
 \beaa
&& \left(6\Up-\frac{5M}{r}+\frac{5Q^2}{r^2} \right) \left( 1-\frac{3M}{r}+\frac{2Q^2}{r^2}  \right) - \Up ^2= \left(6\Up-\frac{5M}{r}+\frac{5Q^2}{r^2} \right) \left( \Up-\frac{M}{r}+\frac{Q^2}{r^2}  \right) - \Up ^2\\
&=&5\Up^2+11\Up \left(-\frac{M}{r}+\frac{Q^2}{r^2} \right) +5\left( -\frac{M}{r}+\frac{Q^2}{r^2}  \right)^2
 \eeaa
 Using \eqref{estimate-Up-r2}, we prove that the above is positive. This shows that condition 5 is verified for $r\geq r_2$, and concludes the proof of the lemma.
\end{proof}

 \subsection{The region $r_1\leq r \leq r_2$}
In the region $r_1 \leq r \leq r_2$, the function $w$ is defined to be $w=\frac{ 2}{ r^2} \Up(r)$.
 From \eqref{eq:John2-1} we then obtain 
  \beaa
 r^2 f&=& \int_{r_P} ^r  \frac{r^2}{\Up} w=\int_{r_P} ^{r}  \frac{r^2}{\Up}  \frac{2}{r^2}  \Up=\int_{r_P} ^{r}  2  =2r-2r_P
 \eeaa
  which gives
\bea
 f&=&\frac{2}{r}-\frac{2r_P}{r^2}=\frac{2(r-r_P)}{r^2} \qquad \text{for $r_1 \leq r \leq r_2$}
\eea
Notice that the above implies 
\beaa
f'&=&-\frac{2}{r^2}+2\frac{2r_P}{r^3}=\frac{-2r+4r_P}{r^3} 
\eeaa
Notice that if $r \leq r_2 \leq 2 r_P$, then $-2r+4r_P>0$, therefore
\bea
f'&>&0 \qquad \text{for $r_1 \leq r \leq r_2$}
\eea
This proves that Condition 2 is verified in this region. 
We now check Conditions 3, 4 and 5.

  \begin{lemma}\label{lemmafirst-cond3} Condition 3 is verified in the region $r_1 \leq r  \leq r_2$, i.e. for all subextremal Reissner-Nordstr\"om spacetimes with $|Q| <  M$ we have 
\beaa
 A_1 \geq 0, \qquad A_2 \geq 0 \qquad  \text{for $r_1 \leq r \leq r_2$}
\eeaa
\end{lemma}
\begin{proof} Since $\lambda \geq 6$, we have that $A_1 \geq B_1$ and $A_2 \geq B_2$, where we recall
\beaa
 B_1&=&\frac{5f}{r^3}\left( 1-\frac{3M}{r}+\frac{2Q^2}{r^2}  \right)  -\frac 1 2 \pr_r\left(\Up V_1\right)f -\frac 1 4   \Box_g  w  \\
 B_2&=&\frac{2f}{r^3}\left( 1-\frac{3M}{r}+\frac{2Q^2}{r^2}  \right)  -\frac 1 2 \pr_r\left(\Up V_2\right)f -\frac 1 4    \Box_g w
 \eeaa 
 We compute $\square_g w$. From $w=\frac {2 }{r^2} \Up(r)$ we obtain
   \beaa
   \pr_r w&=&2\pr_r( \frac {1}{ r^2} \Up)  = 2\left(- 2r^{-3}\Up+ r^{-2}\Up'\right)\\
   &=& 2\left(- 2r^{-3}\left(1-\frac{2M}{r}+\frac{Q^2}{r^2} \right)+ r^{-2}\left(\frac{2M}{r^2}-\frac{2Q^2}{r^2}\right)\right)=-\frac{4}{r^3}\left(1-\frac{3M }{r} +\frac{2 Q^2}{r^2}\right)
   \eeaa
   For a radial function $w=w(r)$ we have $\square_g w=r^{-2} \pr_r(r^2\Up \pr_r w)$ (see Lemma 6 in \cite{Giorgi4}). We therefore compute
   \beaa
    - \pr_r \left( r^2 \Up \pr_rw\right)&=& 4 \pr_r \left( \frac{\Up}{r}\left(1-\frac{3M }{r} +\frac{2 Q^2}{r^2}\right)\right)\\
    &=& 4  \left( \frac{\Up'}{r}\left(1-\frac{3M }{r} +\frac{2 Q^2}{r^2}\right)- \frac{\Up}{r^2}\left(1-\frac{3M }{r} +\frac{2 Q^2}{r^2}\right)+ \frac{\Up}{r}\left(\frac{3M }{r^2} -\frac{4 Q^2}{r^3}\right)\right)\\
    &=& 4  \left( -\frac{1}{r^2}\left(1-\frac{4M}{r}+\frac{3Q^2}{r^2}\right)\left(1-\frac{3M }{r} +\frac{2 Q^2}{r^2}\right)+ \frac{\Up}{r}\left(\frac{3M }{r^2} -\frac{4 Q^2}{r^3}\right)\right)
   \eeaa
We obtain
\beaa
-\frac 1 4  \square_g w&=&  -\frac{1}{r^4}\left(1-\frac{4M}{r}+\frac{3Q^2}{r^2}\right)\left(1-\frac{3M }{r} +\frac{2 Q^2}{r^2}\right)+ \frac{\Up}{r^3}\left(\frac{3M }{r^2} -\frac{4 Q^2}{r^3}\right)
\eeaa
We now evaluate $B_1$ and $B_2$ in this region. For $r_1 \leq r \leq r_2$, we have 
\bea\label{explicit-Up}
\Up= \frac{c}{r^2} \left(Mr-Q^2\right) \qquad \text{ for $\frac 1 2 \leq c \leq 2$}
\eea
where $c=\frac 1 2$ corresponds to $r=r_1$ and $c=2$ corresponds to $r=r_2$. 
We write
\beaa
-\frac 1 4  \square_g w&=&  -\frac{1}{r^8}(c-1)(c-2) \left(Mr-Q^2\right)^2 + \frac{1}{r^8} c \left(Mr-Q^2\right) \left(3M r -4 Q^2\right)\\
&=&  \frac{1}{r^8}\left(-(c-1)(c-2) \left(Mr-Q^2\right) +  c  \left(3M r -4 Q^2\right)\right)\left(Mr-Q^2\right)\\
&=&  \frac{1}{r^8}\left(-(c^2-6c+2) Mr +(c^2-7c+2) Q^2 \right)\left(Mr-Q^2\right)
\eeaa
On the other hand, using \eqref{potential-r21}, we have 
\beaa
-\frac 1 2 \pr_r\left(\Up V_1\right) &=&\frac{1}{r^3}\left(1-\frac{4M}{r}+\frac{3Q^2}{r^2}     \right)\Up+\frac{5Q^2}{r^5} \left(2\Up-\frac{M}{r}+\frac{Q^2}{r^2} \right)\\
&=&\frac{1}{r^3}\left(\frac{c-2}{r^2} \left(Mr-Q^2\right)  \right)\frac{c}{r^2} \left(Mr-Q^2\right)+\frac{5Q^2}{r^5} \left(\frac{2c-1}{r^2} \left(Mr-Q^2\right)\right)\\
&=&\frac{1}{r^7}\left(c(c-2) \left(Mr-Q^2\right)+5Q^2 \left(2c-1\right)\right) \left(Mr-Q^2\right) \\
&=&\frac{1}{r^7}\left(\left(c^2-2c\right) Mr-\left(c^2-12c+5 \right)Q^2\right) \left(Mr-Q^2\right) 
\eeaa
and $\frac{5}{r^3}\left( 1-\frac{3M}{r}+\frac{2Q^2}{r^2}  \right)= \frac{5(c-1)}{r^5} \left(Mr-Q^2\right)$. 
Since $f= \frac{2(r-r_P)}{r^2}$, we have 
\beaa
&&\frac{5f}{r^3}\left( 1-\frac{3M}{r}+\frac{2Q^2}{r^2}  \right)  -\frac 1 2 \pr_r\left(\Up V_1\right)f\\
&=& \frac{1}{r^7}\left(5(c-1)r^2 + \left(c^2-2c\right) Mr-\left(c^2-12c+5 \right)Q^2 \right) \frac{2(r-r_P)}{r^2} \left(Mr-Q^2\right)
\eeaa
and therefore 
\beaa
B_1&=& \frac{1}{r^8}\Big[5 (c-1) r^2\frac{2(r-r_P)}{r}+\left(c(c-2)\frac{2(r-r_P)}{r}-(c^2-6c+2) \right)Mr\\
&&-\left((c^2-12c+5)\frac{2(r-r_P)}{r}-(c^2-7c+2)\right) Q^2 \Big]  \left(Mr-Q^2\right) 
\eeaa
Notice that the above clearly increases as $Q$ decreases. In particular, the extremal case $|Q|=M$ is the worst case scenario in the positivity of the above coefficient. We check that the term in parenthesis is positive at the extremal case, and that will imply that is positive for all subextremal cases $|Q|<M$. 

Notice that equation \eqref{explicit-Up} translates into
\beaa
1-\frac{(2+c)M}{r}+\frac{(1+c)Q^2}{r^2}=0 \qquad \text{ for $\frac 1 2 \leq c \leq 2$}
\eeaa
This gives
\beaa
r&=&\frac{(2+c)M+\sqrt{(2+c)^2M^2-4(1+c)Q^2}}{2}\\
&>&\frac{(2+c)M+\sqrt{(4+4c+c^2)M^2-4(1+c)M^2}}{2}=\frac{(2+2c)M}{2}=(1+c)M
\eeaa
In particular $r=(1+c)M$ at the extremal case $|Q|=M$, where the above becomes 
\beaa
B_1&=& \frac{1}{r^8}\Big[5 (c-1) ((1+c)M)^2\frac{2(c-1)}{c+1}+\frac{1}{c+1}\left(c^4-c^3+27c^2-33c+10 \right) M^2 \Big]  \left(Mr-Q^2\right) \\
&=& \frac{M^2}{(c+1)r^8}\Big[10 (c-1)^2 (1+c)^2+\left(c^4-c^3+27c^2-33c+10 \right)  \Big]  \left(Mr-Q^2\right) \\
&=& \frac{M^2}{(c+1)r^8}\Big[11 c^4-c^3+7c^2-33c+20   \Big]  \left(Mr-Q^2\right) 
\eeaa
which is a positive polynomial for all $c$, and in particular for $\frac 1 2 \leq c \leq 2$.

Similarly we compute $B_2$. Using \eqref{potential-r22}, we have
\beaa
-\frac 1 2   \pr_r\left(\Up V_2\right)&=&\frac{4}{r^3}\left(1-\frac{4M}{r}+\frac{3Q^2}{r^2}\right) \Up+\frac{2Q^2}{r^5} \left(2\Up-\frac{M}{r}+\frac{Q^2}{r^2} \right)\\
&=&\frac{4}{r^3}\left(\frac{c-2}{r^2} \left(Mr-Q^2\right)  \right)\frac{c}{r^2} \left(Mr-Q^2\right)+\frac{2Q^2}{r^5} \left(\frac{2c-1}{r^2} \left(Mr-Q^2\right)\right)\\
&=&\frac{1}{r^7}\left(4c(c-2) \left(Mr-Q^2\right)+2Q^2 \left(2c-1\right)\right) \left(Mr-Q^2\right) \\
&=&\frac{1}{r^7}\left(\left(4c^2-8c\right) Mr-\left(4c^2-12c+2 \right)Q^2\right) \left(Mr-Q^2\right) 
\eeaa
We therefore have
\beaa
B_2&=&  \frac{2 }{r^3}\frac{c-1}{r^2} \left(Mr-Q^2\right)\frac{2(r-r_P)}{r^2}+\frac{1}{r^8}\Big[\left((4c^2-8c)\frac{2(r-r_P)}{r}-(c^2-6c+2) \right)Mr\\
&&-\left((4c^2-12c+2)\frac{2(r-r_P)}{r}-(c^2-7c+2)\right) Q^2 \Big]  \left(Mr-Q^2\right) 
\eeaa
At the extremal case the above becomes
\beaa
B_2&=&\frac{1}{r^8} \Big[  2 (c-1) ((1+c)M)^2\frac{2(c-1)}{c+1}+\frac{1}{c+1}\left(7c^4-19c^3+27c^2-15c+4\right)M^2\Big]  \left(Mr-Q^2\right) \\
&=&\frac{M^2}{(c+1)r^8} \Big[  4 (c-1)^2 (1+c)^2+7c^4-19c^3+27c^2-15c+4\Big]  \left(Mr-Q^2\right) \\
&=&\frac{M^2}{(c+1)r^8} \Big[  11 c^4-19c^3+19c^2-15c+8\Big]  \left(Mr-Q^2\right) 
\eeaa
which is positive for all $c$, and in particular for $\frac 1 2 \leq c \leq 2$. This proves the lemma. 
 \end{proof}

  \begin{lemma}\label{lemmafirst-cond4} Condition 4 is verified in the region $r_1\leq r \leq r_2$, i.e. for all subextremal Reissner-Nordstr\"om spacetimes with $|Q| <  M$ we have 
\beaa
D_1 \geq 0 \qquad \text{ for $r_1\leq r \leq r_2$}
\eeaa
\end{lemma}
\begin{proof} As above, we check the positivity in the extremal case $|Q|=M$. Using the expressions for $B_1$ and $B_2$ obtained in Lemma \ref{lemmafirst-cond3} at the extremal case, we have
\beaa
B_1+B_2&=& \frac{M^2}{(c+1)r^8}\Big[22 c^4-20c^3+26c^2-48c+28 \Big]  \left(Mr-Q^2\right) 
\eeaa
and therefore
\beaa
&& \frac{ f}{r^3}\left( 1-\frac{3M}{r}+\frac{2Q^2}{r^2}  \right) (B_1+B_2)\\
&=&  \frac{(c-1)}{r^5} \frac{M^2}{(c+1)r^9} \frac{2(r-r_P)}{r}\Big[22 c^4-20c^3+26c^2-48c+28 \Big]  \left(Mr-Q^2\right)^2\\
 &=& \frac{M^2}{(c+1)^2r^{14}} 2(c-1)^2\Big[22 c^4-20c^3+26c^2-48c+28 \Big]  \left(Mr-Q^2\right)^2
\eeaa

We now consider the term $- \frac{4 Q^2}{r^8} f^2 \left(1-\frac{3M}{r}+\frac{2Q^2}{r^2} +\frac 1 2 \Up \right)^2$. We have
\beaa
- \frac{4 Q^2}{r^8} f^2 \left(1-\frac{3M}{r}+\frac{2Q^2}{r^2} +\frac 1 2 \Up \right)^2&=& - \frac{4 Q^2}{r^8} f^2 \left(\frac{c-1}{r^2} \left(Mr-Q^2\right)+\frac 1 2\frac{c}{r^2} \left(Mr-Q^2\right) \right)^2\\
&=& - \frac{2 Q^2}{r^8} \left(\frac{2(r-r_P)}{r^2} \right)^2\frac{(3c-2)^2}{r^4} \left(Mr-Q^2\right)^2\\
&=& - \frac{2 Q^2}{r^{14}} \left(\frac{2(r-r_P)}{r} \right)^2(3c-2)^2 \left(Mr-Q^2\right)^2
\eeaa
At the extremal case the above becomes
\beaa
&=& - \frac{2 M^2}{r^{14}} \left(\frac{2(c-1)}{c+1} \right)^2(3c-2)^2 \left(Mr-Q^2\right)^2= - \frac{8 M^2}{(c+1)^2r^{14}}(c-1)^2(3c-2)^2 \left(Mr-Q^2\right)^2
\eeaa
We therefore obtain the polynomial 
\beaa
\Big[22 c^4-20c^3+26c^2-48c+28 \Big]  - 4(3c-2)^2 =22c^4-20c^3-10c^2+12
\eeaa
which is positive for all $c$, and in particular for $\frac 1 2 \leq c \leq 2$.

\end{proof}

 \begin{lemma}\label{lemmafirst-cond5} Condition 5 is verified in the region $r_1 \leq r  \leq r_2$, i.e. for all subextremal Reissner-Nordstr\"om spacetimes with $|Q| < M$ we have 
\beaa
D_2 \geq 0 \qquad \text{ for $r_1\leq r \leq r_2$}
\eeaa
\end{lemma}
\begin{proof} As above, we check the positivity in the extremal case $|Q|=M$. Using the expressions for $B_1$ and $B_2$ obtained in Lemma \ref{lemmafirst-cond3} at the extremal case, we have
\beaa
B_1B_2&=& \frac{M^4}{(c+1)^2r^{16}}\left(11 c^4-c^3+7c^2-33c+20   \right) \left(  11 c^4-19c^3+19c^2-15c+8\right) \left(Mr-Q^2\right)^2
\eeaa
We now consider the term $- \frac{16 Q^2}{r^8} f^2 \left(1-\frac{3M}{r}+\frac{2Q^2}{r^2} +\frac 1 2 \Up \right)^2$. We have
\beaa
- \frac{16 Q^2}{r^8} f^2 \left(1-\frac{3M}{r}+\frac{2Q^2}{r^2} +\frac 1 2 \Up \right)^2&=& - \frac{8 Q^2}{r^{14}} \left(\frac{2(r-r_P)}{r} \right)^2(3c-2)^2 \left(Mr-Q^2\right)^2
\eeaa
At the extremal case the above becomes
\beaa
&=& - \frac{8 M^2}{r^{14}} \left(\frac{2(c-1)}{c+1} \right)^2(3c-2)^2 \left(Mr-Q^2\right)^2\\
&=& - \frac{32 M^2}{(c+1)^2r^{14}}(c-1)^2(3c-2)^2 \left(Mr-Q^2\right)^2= - \frac{ M^4}{(c+1)^2r^{16}} 32(c+1)^2(c-1)^2(3c-2)^2 \left(Mr-Q^2\right)^2
\eeaa
We therefore obtain the polynomial 
\beaa
&&\left(11 c^4-c^3+7c^2-33c+20   \right) \left(  11 c^4-19c^3+19c^2-15c+8\right)-32(c+1)^2(c-1)^2(3c-2)^2\\
&=&121c^8-220c^7+17c^6-296c^5+1531c^4-1888c^3+899c^2-180c+32
\eeaa
which is positive for all $c$, and in particular for $\frac 1 2 \leq c \leq 2$. This concludes the proof of the lemma.

\end{proof}

\section{$r^p$-hierarchy estimates}\label{rp-sec}

In this section, we conclude the proof of Theorem \ref{main-theorem} by obtaining the $r^p$-hierarchy estimates for the mixed spin $\pm1$ and spin $\pm2$ system of equations.  The equations in the system have right hand sides which present good decay in $r$, and therefore the $r^p$-estimates for the system do not differ from the standard derivation for a single wave equation as in \cite{rp}.

  To derive the $r^p$-estimates, we apply Lemma \ref{divergence-combined-current}  to the vector field $Z=l(r) e_4$, for the function $l(r)=r^p$.

 \begin{proposition}\label{main-identity-r^pvf}  Consider a vectorfield $Z=l(r) e_4$, a pair of scalar functions $w=(w_1, w_2)$, with $w_1=w_2=\frac{2l(r)}{r}$, and a pair of one forms $M=(M_1, M_2)$ with $M_1=M_2=\frac{ 2l'}{r}  e_4=\frac{2l'}{rl} Z$. Let $\Phi_1$ and $\Phi_2$ be a 1-tensor and a symmetric traceless 2-tensor respectively, satisfying the system of coupled wave equations \eqref{box-phi-1} and \eqref{box-phi-2}. Then we have 
\beaa
\EE^{(Z, w, M)}[\Phi_1, \Phi_2] &=_s&\frac 1 2 l' |\check{\nabb}_4 \Phi_1|^2+\frac 1 2 \left(- l' +\frac 2 r l\right)(|\nabb \Phi_1|^2+V_1 |\Phi_1|^2)\\
&&+4 l' |\check{\nabb}_4 \Phi_2|^2+4 \left(- l' +\frac 2 r l\right)(|\nabb \Phi_2|^2+ V_2 |\Phi_2|^2)+\mbox{err}[\Phi_1, \Phi_2]
\eeaa
where
 \beaa
 \mbox{err}[\Phi_1, \Phi_2]&=&O\left(\frac{M+Q}{r^2}, \frac{Q^2}{r^3}\right) \left( |l|+r|l'| +r^2|l''| \right)\Big(|\nabb_4 \Phi_1|^2+|\nabb_4\Phi_2|^2 +r^{-2}(|\Phi_1|^2+|\Phi_2|^2)\Big)
 \eeaa
 \end{proposition}
 \begin{proof} We start by computing $\EE^{(Z, w, 0)}[\Phi_1, \Phi_2]$. 
 Using \eqref{definition-EE} we have 
 \beaa
 \EE^{(Z, w, 0)}[\Phi_1, \Phi_2]   &=&\frac 1 2 \TT[\Phi_1]  \c\piZ+\left( - \frac 1 2 Z( V_1 ) -\frac 1 4   \Box_g  w_1 \right)|\Phi_1|^2+\frac 12  w_1 \LL_1[\Phi_1]     \\
 &&+4 \TT[\Phi_2]  \c\piZ+\left( - 4 Z( V_2 ) -2   \Box_g  w_2 \right)|\Phi_2|^2+4  w_2 \LL_2[\Phi_2]      \\
  &&+\frac{ 4Q}{r^2}\left(w_1+w_2 +\frac{4}{r}  Z(r) - \tr \piZ\right)\DDs_2\Phi_1 \c  \Phi_2  -\frac{8Q}{r^2}([Z, \DDs_2] \Phi_1) \c \Phi_2  
 \eeaa
We have for $Z=l(r) e_4$ (see Corollary 4 in \cite{Giorgi4})
  \beaa
 \TT[\Phi_1] \c \piZ&=& \left(- l' +\frac 2 r l\right)|\nabb \Phi_1|^2+ \left( \Up l'+\left(-\frac{2M}{r^2}+\frac{2Q^2}{r^3} \right)l \right)|\nabb_4 \Phi_1|^2-\frac{2}{r}l \LL_1[\Phi_1]- l' V_1 |\Phi_1|^2
 \eeaa
and similarly for $\Phi_2$.
We therefore obtain
  \beaa
 \EE^{(Z, w, 0)}[\Phi_1, \Phi_2]   &=&\frac 1 2 \left(- l' +\frac 2 r l\right)(|\nabb \Phi_1|^2+V_1 |\Phi_1|^2)+ \frac 1 2 \left( \Up l'+\left(-\frac{2M}{r^2}+\frac{2Q^2}{r^3} \right)l \right)|\nabb_4 \Phi_1|^2\\
 &&+\left( - \frac 1 2 Z( V_1 ) -\frac 1 4   \Box_g  w_1 - \frac l r  V_1\right)|\Phi_1|^2+\left(\frac 12  w_1   -\frac{l}{r}\right) \LL_1[\Phi_1]  \\
 &&+4 \left(- l' +\frac 2 r l\right)(|\nabb \Phi_2|^2+ V_2 |\Phi_2|^2)+ 4\left( \Up l'+\left(-\frac{2M}{r^2}+\frac{2Q^2}{r^3} \right)l \right)|\nabb_4 \Phi_2|^2\\
 &&+\left( - 4 Z( V_2 ) -2   \Box_g  w_2 - \frac 8 r l  V_2\right)|\Phi_2|^2+\left(4  w_2   -\frac{8}{r}l \right)\LL_2[\Phi_2]    \\
  &&+\frac{ 4Q}{r^2}\left(w_1+w_2 +\frac{4}{r}  Z(r) - \tr \piZ\right)\DDs_2\Phi_1 \c  \Phi_2  -\frac{8Q}{r^2}([Z, \DDs_2] \Phi_1) \c \Phi_2  
 \eeaa
With the choice $w=w_1=w_2=\frac{2l}{r}$, 
 the terms involving $\LL_1[\Phi_1]$ and $\LL_2[\Phi_2]$, cancel out. Using again the formula $\square_g w=r^{-2} \pr_r(r^2\Up \pr_r w)$ (see Lemma 6 in \cite{Giorgi4}), we obtain 
 \beaa
 \Box_g w &=& \frac{2l''}{r}+O\left(\frac{M}{r^4}, \frac{Q^2}{r^5}\right) \left[ |l|+r|l'|+r^2|l''|  \right]
 \eeaa 
Therefore
\beaa
- \frac 1 2 Z( V_1 ) -\frac 1 4   \Box_g  w_1 - \frac l r  V_1&=&-  \frac{l''}{2r}- \frac l 2 \left( V_1' + \frac 2 r  V_1\right) +O\left(\frac{M}{r^4}, \frac{Q^2}{r^5}\right) \left[ |l|+r|l'|+r^2|l''|  \right]  \\
 - 4 Z( V_2 ) -2   \Box_g  w_2 - \frac 8 r l  V_2&=& -  8\frac{l''}{2r}- 4 l  \left( V_2' + \frac 2 r  V_2\right) +O\left(\frac{M}{r^4}, \frac{Q^2}{r^5}\right) \left[ |l|+r|l'|+r^2|l''|  \right]
\eeaa
Using \eqref{potentials-derivatives}, we have in both cases
\beaa
V_1' + \frac 2 r  V_1&=& O\left(\frac{M}{r^4}, \frac{Q^2}{r^5}\right) \qquad V_2' + \frac 2 r  V_2= O\left(\frac{M}{r^4}, \frac{Q^2}{r^5}\right)
\eeaa
Recall that  $\tr\pi^{(4)}= \frac 4 r$ and  $\tr\piZ= 2l'+\frac{4}{r}l $ (see Corollary 3 in \cite{Giorgi4}), therefore
\beaa
[Z, \DDs_2] \Phi_1&=& l[ e_4, \DDs_2] \Phi_1=-  \frac l r  \DDs_2 \Phi_1\\
w_1+w_2 +\frac{4}{r}  Z(r) - \tr \piZ&=& \frac{4l}{r} +\frac{4}{r}  l - (2l'+\frac{4}{r}l)=- 2l'+\frac{4l}{r} 
\eeaa
Using \eqref{dot-product-estimate-l} to write
 \beaa
\frac{ 4Q}{r^2}\left(w_1+w_2 +\frac{4}{r}  Z(r) - \tr \piZ\right)\DDs_2\Phi_1 \c  \Phi_2 &=& \frac{ 4Q}{r^2}\left(- 2l'+\frac{6l}{r}\right)\DDs_2\Phi_1 \c  \Phi_2 \\
&=_s& O\left(\frac{Q}{r^4}\right) \left[ |l|+r|l'|\right](|\Phi_1|^2+|\Phi_2|^2)
 \eeaa
we conclude
  \beaa
 \EE^{(Z, w, 0)}[\Phi_1, \Phi_2]   &=&\frac 1 2 \left(- l' +\frac 2 r l\right)(|\nabb \Phi_1|^2+V_1 |\Phi_1|^2)+ \frac 1 2  l' |\nabb_4 \Phi_1|^2-  \frac{l''}{2r} |\Phi_1|^2  \\
 &&+4 \left(- l' +\frac 2 r l\right)(|\nabb \Phi_2|^2+ V_2 |\Phi_2|^2)+ 4  l'|\nabb_4 \Phi_2|^2 -  4\frac{l''}{r}|\Phi_2|^2 +\mbox{err}[\Phi_1, \Phi_2]
 \eeaa
 where
 \beaa
 \mbox{err}[\Phi_1, \Phi_2]&=&O\left(\frac{M+Q}{r^2}, \frac{Q^2}{r^3}\right) \left( |l|+r|l'| +r^2|l''| \right)\Big(|\nabb_4 \Phi_1|^2+|\nabb_4\Phi_2|^2 +r^{-2}(|\Phi_1|^2+|\Phi_2|^2)\Big)
 \eeaa
 Using \eqref{definition-EE} we have 
\beaa
\EE^{(Z, w, M)}[\Phi_1, \Phi_2] &=&\EE^{(Z, w, 0)}[\Phi_1, \Phi_2] \\
&&+\frac 1 4( \div M_1)|\Phi_1|^2+\frac 12 \Phi_1 \c M_1(\Phi_1)+2( \div M_2)|\Phi_2|^2+4 \Phi_2 \c M_2(\Phi_2)
\eeaa
Let  $M:=M_1=M_2=\frac{ 2l'}{r}  e_4=\frac{2l'}{rl} Z$, we compute
\beaa
\div M&=&g^{\mu\nu}D_\nu\left(\frac{2l'}{rl} Z_\mu\right)=  \frac{l'}{r l}  \tr\piZ+Z\left(  \frac{2l'}{r l}\right)=  \frac{l'}{r l}  (2l'+\frac{4l}{r})+2l e_4\left(  \frac{l'}{r l}\right)=\frac{2 l'}{r^2} +\frac{2l''}{r}
\eeaa
We deduce
\beaa
\EE^{(Z, w, M)}[\Phi_1, \Phi_2] &=_s&\frac 1 2 \left(- l' +\frac 2 r l\right)(|\nabb \Phi_1|^2+V_1 |\Phi_1|^2)+ \frac 1 2  l'|\nabb_4 \Phi_1|^2   +\frac 1 2\frac{ l'}{r^2} |\Phi_1|^2+ \Phi_1 \c \frac{ l'}{r}  \nabb_4(\Phi_1)\\
 &&+4 \left(- l' +\frac 2 r l\right)(|\nabb \Phi_2|^2+ V_2 |\Phi_2|^2)+ 4  l'|\nabb_4 \Phi_2|^2  + \frac{4 l'}{r^2} |\Phi_2|^2  +8 \Phi_2 \c \frac{ l'}{r}  \nabb_4(\Phi_2)\\
  &&+\mbox{err}[\Phi_1, \Phi_2]
\eeaa
Writing 
\beaa
 \frac 1 2 l'|\nabb_4\Phi_1|^2+ \frac{ l'}{2r^2}|\Phi_1|^2+  r^{-1} l'     \Phi \c \nabb_4(\Phi_1)=\frac  1 2  l'( \nabb_4(\Phi_1)+ r^{-1} \Phi_1)^2=\frac 1 2 l' |\check{\nabb}_4 \Phi_1|^2
 \eeaa
where we recall that $\check{\nabb}_4\Phi_1:=\nabb_4\Phi_1+\frac 1 r \Phi_1$,  and similarly for $\Phi_2$. This concludes the proof of the proposition.

\end{proof}

We finally relate the bulk $\EE^{(Z, w, M)}[\Phi_1, \Phi_2]$  with the weighted bulk norm $\MM_{p\,; \,R}[\Phi_1, \Phi_2]$ in the far away region. 
 For $l(r)=r^p$, the above Proposition implies
  \beaa
\EE^{(Z, w, M)}[\Phi_1, \Phi_2] &=_s&\frac p 2 r^{p-1} |\check{\nabb}_4 \Phi_1|^2+\frac 1 2 (2-p) r^{p-1}(|\nabb \Phi_1|^2+V_1 |\Phi_1|^2)\\
 &&+4 p r^{p-1} |\check{\nabb}_4 \Phi_2|^2+4 (2-p)r^{p-1}(|\nabb \Phi_2|^2+ V_2 |\Phi_2|^2)+\mbox{err}[\Phi_1, \Phi_2]
\eeaa

Given a fixed $\de>0$, for all $\de \le p\le 2-\de$  and $ R\gg \max(\frac{M+Q}{\de}, \frac{Q^2}{\de^2}) $, while integrating in $r \geq R$, the term $\mbox{err}[\Phi_1, \Phi_2]$ can be absorbed by the first two lines above. Thus, we obtain
  \beaa
 \int_{\MM_{\ge R}(0,\tau)  }\EE^{(Z, w, M)}[\Phi_1, \Phi_2] &\ge &\frac 1  4  \int_{\MM_{\ge R}(\tau_1,\tau_2)  } r^{p-1}(p |\check{\nabb}_4 \Phi_1|^2+ (2-p) (|\nabb \Phi_1|^2+ r^{-2} |\Phi_1|^2))\\
 &&+2  \int_{\MM_{\ge R}(\tau_1,\tau_2)  }  r^{p-1}(p |\check{\nabb}_4 \Phi_2|^2+ (2-p)(|\nabb \Phi_2|^2+ r^{-2} |\Phi_2|^2))
\eeaa
Recalling the definition of the spacetime energy \eqref{Morawetz-far-away}, we obtain
 \bea\label{estimate-bulk-rp}
 \int_{\MM_{\ge R}(0,\tau)  }\EE^{(Z, w, M)}[\Phi_1, \Phi_2] &\ges &  \MM_{p\,; \,R}[\Phi_1, \Phi_2](0, \tau)
 \eea
 Consider now the current  $\PP^{(Z, w, M)}[\Phi_1, \Phi_2]$ associated to the vector field $Z$. It is given by
\beaa
\PP^{(Z, w, M)}_\mu[\Phi_1, \Phi_2]&=&  \TT_{\mu\nu}[\Phi_1]Z^\nu+8\TT_{\mu\nu}[\Phi_2]Z^\nu-\frac{8Q}{r^2} l \left(\DDs_2 \Phi_1 \c \Phi_2 \right)  g(e_4, e_\mu)\\
&&+\frac 1 2  w \Phi_1 D_\mu \Phi_1 -\frac 1 4   \pr_\mu w  |\Phi_1|^2+\frac 1 2  \frac{ l'}{r} g(e_4, e_\mu) |\Phi_1|^2\\
&&+4  w \Phi_2 D_\mu \Phi_2 -2   \pr_\mu w  |\Phi_2|^2+4  \frac{ l'}{r} g(e_4, e_\mu)  |\Phi_2|^2
\eeaa
For the boundary terms we compute
\beaa
\PP^{(Z, w, M)}[\Phi_1, \Phi_2]\c   e_4&=& l\TT[\Phi_1]_{44}+ \frac l r   \Phi_1 \c \nabb_4\Phi_1 - \frac 1 2 e_4(r^{-1} l)| \Phi_1|^2\\
&&+8l\TT[\Phi_2]_{44}+ 8\frac l r   \Phi_2 \c \nabb_4\Phi_2 - 4 e_4(r^{-1} l)| \Phi_2|^2\\
&=&l|\nabb_4\Phi_1+ \frac 1 r   \Phi_1|^2 - \frac 1 2 r^{-2}  e_4( r  l |\Phi_1|^2)+8l|\nabb_4\Phi_2+ \frac 1 r   \Phi_2|^2 - 4 r^{-2}  e_4( r  l |\Phi_2|^2)\\
&=&l|\check{\nabb}_4\Phi_1|^2 - \frac 1 2 r^{-2}  e_4( r  l |\Phi_1|^2)+8l|\check{\nabb}_4\Phi_2|^2 - 4 r^{-2}  e_4( r  l |\Phi_2|^2)
\eeaa
and 
\beaa
\PP^{(Z, w, M)}[\Phi_1, \Phi_2]\c   e_3&=&l\TT[\Phi_1]_{34} +  \frac 1 2 r^{-1}   l e_3(|\Phi_1|^2) -\frac 1  2  e_3  ( r^{-1}   l  ) |\Phi_1|^2- r^{-1} l' |\Phi_1|^2\\
&&+ 8l\TT[\Phi_2]_{34}+  4 r^{-1}   l e_3(|\Phi_2|^2) -4  e_3  ( r^{-1}   l  ) |\Phi_2|^2- 8r^{-1} l' |\Phi_2|^2\\
&&+\frac{16Q}{r^2} l \left(\DDs_2 \Phi_1 \c \Phi_2 \right)\\
&=&l(|\nabb\Phi_1|^2+V_1 |\Phi_1|^2)+ 8l(|\nabb\Phi_2|^2+V_2 |\Phi_2|^2)\\
&& +\frac  12  r^{-2}e_3\big ( r  l |\Phi_1|^2) +4  r^{-2}e_3\big ( r  l |\Phi_2|^2)+\mbox{err}[\Phi_1, \Phi_2]
\eeaa

We are now ready to derive the $r^p$-estimates for the mixed spin $\pm1$ and spin $\pm2$ system. 
\begin{proposition}
  \label{theorem:Daf-Rodn1} Let $\Phi_1$ and $\Phi_2$ be a one form and a symmetric traceless two tensor respectively, satisfying the system of coupled wave equations \eqref{box-phi-1} and \eqref{box-phi-2}.  Consider  a fixed $\de>0$ and let $R\gg \max(\frac{M+Q}{\de}, \frac{Q^2}{\de^2}) $. Then for all $\de\le p\le 2-\de $  the following $r^p$-estimates hold: 
       \bea\label{estimates=rp}
   E_{p, R}[\Phi_1, \Phi_2](\tau)   +\MM_{p\,; \,R}[\Phi_1, \Phi_2](0, \tau)&\les&E_{p}[\Phi_1, \Phi_2](0) 
     \eea
\end{proposition}
  \begin{proof}  Let   $\th=\th(r) $ supported    for   $r\ge R/2$   with     $\th=1$ for $r\ge R$  such that $l_p=\th(r) r^p$, $Z_p=l_p e_4$, $w_p=\frac{2l_p}{r}$, $M_p=\frac{ 2l_p'}{r}  e_4$. We apply the divergence theorem to $\PP_p:=\PP^{(Z_p, w_p, M_p)}[\Phi_1, \Phi_2]$ in the spacetime region bounded by $\Sigma_0$ and $\Sigma_\tau$. Using \eqref{divergence-P-EE}, by divergence theorem we have
     \beaa
   \int_{\Si_\tau}   \PP_p\c  N_\Si +\int_{\mathcal{I}^+(0, \tau)}   \PP_p\c e_3+\int_{\MM(0, \tau)}\EE^{(Z_p, w_p, M_p)}[\Phi_1, \Phi_2]=\int_{\Si_0}   \PP_p\c   N_\Si 
     \eeaa
 Recall that $l_p$ vanishes for $r \le R/2$.    We can estimate some of the terms as follows:
     \beaa
 &&    \bigg|\int_{\Si_{R/2 \leq r \leq R}(\tau)}   \PP_p\c   N_\Si\bigg|\les R^p E[\Phi_1, \Phi_2](\tau), \qquad \bigg|\int_{\Si_{R/2 \leq r \leq R}(0)}   \PP_p\c   N_\Si\bigg|\les R^p E[\Phi_1, \Phi_2](0), \\
 &&   \bigg|  \int_{\MM_{R/2 \leq r \leq R}(0, \tau)}\EE^{(Z_p, w_p, M_p)}[\Phi_1, \Phi_2]\bigg|\les  R^{p-1} \Mor[\Phi_1, \Phi_2](0,\tau),
     \eeaa
          where recall that $\Mor[\Phi_1, \Phi_2](0,\tau)$ is defined by \eqref{def:Mor-bulk}.
     Hence,
      \beaa
    &&\int_{\Si_{r \geq R}(\tau)}   \PP_p\c  N_\Si  +\int_{\mathcal{I}^+(0, \tau)}   \PP_p\c e_3+\int_{\MM_{r \geq R}(0, \tau)}\EE^{(Z_p, w_p, M_p)}[\Phi_1, \Phi_2]\\
    &\les&\int_{\Si_{r \geq R}(0)}   \PP_p\c   N_\Si +R^p\left(E[\Phi_1, \Phi_2](0)+E[\Phi_1, \Phi_2](\tau)+R^{-1} \Mor[\Phi_1, \Phi_2](0,\tau)\right) 
     \eeaa
  Using the expressions for $\PP \c e_4$ and $\PP \c e_3$, we bound
 \beaa
 \int_{\Si_{r \geq R}(\tau)}   \PP\c  N_\Si&=& \int_{\Si_{r \geq R}(\tau)}   \PP\c e_4\\
 &=&\int_{\Si_{r \geq R}(\tau)} r^p|\check{\nabb}_4\Phi_1|^2 - \frac 1 2 r^{-2}  e_4(   r^{p+1} |\Phi_1|^2)+8r^p|\check{\nabb}_4\Phi_2|^2 - 4 r^{-2}  e_4(  r^{p+1} |\Phi_2|^2)\\
 &\ges&   E_{p, R}[\Phi_1, \Phi_2](\tau) 
 \eeaa
  by performing the integration by parts   for the second terms in the integrals, and absorbing the boundary term. Also,
 \beaa
  \int_{\mathcal{I}^+(0, \tau)}   \PP\c e_3&=&  \int_{\mathcal{I}^+(0, \tau)} l(|\nabb\Phi_1|^2+V_1 |\Phi_1|^2) + 8l(|\nabb\Phi_2|^2+V_2 |\Phi_2|^2) \\
  &&+ \int_{\mathcal{I}^+(0, \tau)}\frac  12  r^{-2}e_3\big ( r  l |\Phi_1|^2) +4  r^{-2}e_3\big ( r  l |\Phi_2|^2)+\mbox{err}[\Phi_1, \Phi_2]\\
  &\ges&\int_{\mathcal{I}^+(0, \tau)} r^p|\nabb \Phi_1|^2+r^{p-2}|\Phi_1|^2+r^p|\nabb \Phi_2|^2+r^{p-2}|\Phi_2|^2
 \eeaa
     Using \eqref{estimate-bulk-rp}, we obtain 
       \beaa
   E_{p, R}[\Phi_1, \Phi_2](\tau)   +\MM_{p\,; \,R}[\Phi_1, \Phi_2](0, \tau)&\les&E_{p}[\Phi_1, \Phi_2](0) \\
   && +R^p\left(E[\Phi_1, \Phi_2](\tau)+R^{-1} \Mor[\Phi_1, \Phi_2](0,\tau)\right) 
     \eeaa 
By combining the above with the boundedness of the energy in Proposition \ref{prop-energy-deg} and with the Morawetz estimates in \eqref{final-morawets}, we prove the proposition.  
         \end{proof}     
         
         By summing \eqref{estimates=rp} and \eqref{final-morawets} we obtain the boundedness of the weighted energy \eqref{main-1} and the integrated weighted estimates \eqref{main-2}, therefore proving Theorem \ref{main-theorem}.

\begin{acknowledgements}
 The author is grateful to Pei-Ken Hung and Yakov Shlapentokh-Rothman for helpful discussions.  
\end{acknowledgements}



\end{document}